\documentclass[]{easychair}

\usepackage{graphicx}
\usepackage{amssymb}
\usepackage{amsmath}
\usepackage{mathtools}
\usepackage{paralist}
\usepackage{todonotes}
\usepackage{txfonts}
\usepackage{relsize}

\newtheorem{definition}{Definition}
\newtheorem{lemma}{Lemma}
\newtheorem{theorem}{Theorem}

\newtheorem{fact}{Fact}
\newtheorem{claim}{Claim}
\newtheorem{proposition}{Proposition}
\newtheorem{assumption}{Assumption}

\newcommand{\nat}{\mathbb{N}}

\newcommand{\arity}{\#}
\newcommand{\arityof}[1]{{\arity{#1}}}

\newcommand{\cardof}[1]{||{#1}||}

\newcommand{\width}[1]{\mathrm{wd}({#1})}
\newcommand{\twd}[1]{\mathrm{twd}({#1})}
\newcommand{\adh}[3]{\mathrm{adh}_{#1}({#2},{#3})}

\newcommand{\univ}{\mathsf{U}}
\newcommand{\universeOf}[1]{\mathsf{#1}}

\newcommand{\vars}{\mathcal{V}}
\newcommand{\Vars}{\mathcal{X}}

\newcommand{\isdef}{\stackrel{\scalebox{0.5}{$\mathsf{def}$}}{=}}

\newcommand{\interv}[2]{[{#1},{#2}]}
\newcommand{\tuple}[1]{\langle {#1} \rangle}

\newcommand{\set}[1]{\{ {#1} \}}

\newcommand{\pow}[1]{2^{#1}}

\newcommand{\dom}[1]{\mathrm{dom}({#1})}

\newcommand{\np}{$\mathsf{NP}$}
\newcommand{\ptime}{$\mathsf{PTIME}$}

\newcommand{\signature}{\mathcal{F}}
\newcommand{\trianglesign}{\signature^\basictriangle}
\newcommand{\trianglefansign}{\signature^{\fan{}}}

\newcommand{\hrsignature}{\signature^{\mathsf{HR}}}

\newcommand{\relations}{\mathcal{R}}

\newcommand{\arel}{\mathsf{r}}
\newcommand{\qrel}{\mathsf{q}}

\newcommand{\arrow}[2]{\xrightarrow{{\scriptscriptstyle #1}}_{{\scriptstyle #2}}}

\newcommand{\auto}[2]{\mathcal{A}_{
    {#1}
    \ifthenelse{\equal{#2}{}}{}{,{#2}}
}}

\newcommand{\autsat}[2]{\mathcal{A}^{\scriptscriptstyle\mathsf{sat}}_{
    {#1}
    \ifthenelse{\equal{#2}{}}{}{,{#2}}
}}

\newcommand{\autcut}[2]{\mathcal{A}^{\scriptscriptstyle\mathsf{cut}}_{
    {#1}
    \ifthenelse{\equal{#2}{}}{}{,{#2}}
}}

\newcommand{\autcst}[2]{\mathcal{A}^{\scriptscriptstyle\mathsf{cst}}_{
    {#1}
    \ifthenelse{\equal{#2}{}}{}{,{#2}}
  }
}

\makeatletter
\newcommand*{\da@rightarrow}{\mathchar"0\hexnumber@\symAMSa 4B }
\newcommand*{\da@leftarrow}{\mathchar"0\hexnumber@\symAMSa 4C }
\newcommand*{\xdashrightarrow}[2][]{%
  \mathrel{%
    \mathpalette{\da@xarrow{#1}{#2}{}\da@rightarrow{\,}{}}{}%
  }%
}
\newcommand{\xdashleftarrow}[2][]{%
  \mathrel{%
    \mathpalette{\da@xarrow{#1}{#2}\da@leftarrow{}{}{\,}}{}%
  }%
}
\newcommand*{\da@xarrow}[7]{%
  \sbox0{$\ifx#7\scriptstyle\scriptscriptstyle\else\scriptstyle\fi#5#1#6\m@th$}%
  \sbox2{$\ifx#7\scriptstyle\scriptscriptstyle\else\scriptstyle\fi#5#2#6\m@th$}%
  \sbox4{$#7\dabar@\m@th$}%
  \dimen@=\wd0 %
  \ifdim\wd2 >\dimen@
    \dimen@=\wd2 %
  \fi
  \count@=2 %
  \def\da@bars{\dabar@\dabar@}%
  \@whiledim\count@\wd4<\dimen@\do{%
    \advance\count@\@ne
    \expandafter\def\expandafter\da@bars\expandafter{%
      \da@bars
      \dabar@
    }%
  }%
  \mathrel{#3}%
  \mathrel{%
    \mathop{\da@bars}\limits
    \ifx\\#1\\%
    \else
      _{\copy0}%
    \fi
    \ifx\\#2\\%
    \else
      ^{\copy2}%
    \fi
  }%
  \mathrel{#4}%
}
\makeatother

\newcommand{\store}{\mathfrak{s}}
\newcommand{\struc}{\sigma}

\newcommand{\cardconstr}[3]{\mathsf{card}_{{#2},{#3}}({#1})}
\newcommand{\astruc}{\mathcal{S}}

\newcommand{\false}{\mathrm{false}}

\DeclareMathOperator*{\Pop}{\mathop{\scalebox{1.5}{\raisebox{-0.2ex}{$\parallel$}}}\hspace*{1pt}}%

\newcommand{\cmso}{$\mathsf{CMSO}$}

\newcommand{\mso}{$\mathsf{MSO}$}

\newcommand{\Models}{\models}
\newcommand{\edgrel}{\mathsf{edge}}
\newcommand{\srcrel}[1]{\mathsf{source}_{#1}}

\renewcommand{\mod}{~\mathrm{mod}~}

\newcommand{\sintfusion}[2]{\widetilde{\mathtt{IF}}({#1}\ifthenelse{\equal{#2}{}}{}{,{#2}})}

\newcommand{\step}[1]{\Rightarrow_{\scriptscriptstyle{#1}}}

\newcommand{\langof}[2]{\mathcal{L}^{#1}({#2})}

\newcommand{\proj}[2]{{#1}\!\!\downharpoonleft_{\scriptscriptstyle{#2}}}

\newcommand{\graphof}[3]{\mathrm{subgraph}_{#1}[{#2}
\ifthenelse{\equal{#3}{}}{]}{,{#3}]}}

\newcommand{\graph}{G}
\newcommand{\hgraph}{H}
\newcommand{\kgraph}{K}

\newcommand{\strucof}[1]{\astruc({#1})}
\newcommand{\graphsof}[1]{{\mathcal{G}^{\scriptscriptstyle{#1}}}}
\newcommand{\graphs}{\graphsof{}}

\newcommand{\vertices}{V}
\newcommand{\vertof}[1]{\vertices_{\scriptscriptstyle{#1}}}
\newcommand{\sep}[2]{\mathcal{S}({#1},{#2})}

\newcommand{\edgeof}[1]{\edges_{\scriptscriptstyle{#1}}}

\newcommand{\sources}{\xi}
\newcommand{\sourceof}[1]{\sources_{\scriptscriptstyle{#1}}}

\newcommand{\algof}[1]{\mathcal{#1}}

\newcommand{\emptygraph}{\therefore}

\newcommand{\grammar}{\Gamma}

\newcommand{\period}[2]{\pi_{#2}\ifthenelse{\equal{#1}{}}{}{({#1})}}

\newcommand{\base}[2]{\beta_{#2}\ifthenelse{\equal{#1}{}}{}{({#1})}}

\newcommand{\hval}{\mathbf{val}}

\newcommand{\node}{\mathsf{node}}

\newcommand{\parent}{\mathsf{parent}}
\newcommand{\trans}{\delta}
\newcommand{\scheme}{\Theta}
\newcommand{\defdof}[2]{\mathrm{def}^{{#2}}_{{#1}}}
\newcommand{\defd}[1]{\mathrm{def}_{#1}}
\newcommand{\hr}{$\mathsf{HR}$}
\newcommand{\vr}{$\mathsf{VR}$}

\newcommand{\rules}{\mathcal{R}}

\newcommand{\nonterm}{\mathcal{N}}
\newcommand{\axioms}{\mathcal{X}}

\newcommand{\edges}{{E}}

\newcommand{\tree}{T}
\newcommand{\bag}{\beta}

\newcommand{\twof}[1]{\mathrm{tw}({#1})}

\newcommand{\unit}[2]{\overline{1}_{
        {#1}
        \ifthenelse{\equal{#2}{}}{}{,{#2}}
}}

\newcommand{\zero}[2]{\overline{0}_{
        {#1}
        \ifthenelse{\equal{#2}{}}{}{,{#2}}
}}

\newcommand{\vertex}{\bullet}

\newcommand{\pop}{\parallel}
\newcommand{\fan}[1]{\triangledown^{#1}}
\newcommand{\basictriangle}{\triangle}

\newcommand{\trianglegraphs}{\graphs^\basictriangle}
\newcommand{\trianglefangraphs}{\graphs^{\fan{}}}
\newcommand{\spgraphs}{\graphs^{\mathit{sp}}}
\newcommand{\trianglealg}{\mathlarger{\mathlarger{\basictriangle}}}
\newcommand{\fanalg}{\mathlarger{\mathlarger{\fan{}}}}

\newcommand{\sop}{\circ}

\newcommand{\rename}[1]{\mathsf{rename}_{{#1}}}

\newcommand{\forget}[1]{\mathsf{forget}_{{#1}}}
\newcommand{\bridge}{\multimapboth}

\newcommand*{\langu}{\mathcal{L}}
\newcommand*{\klangu}{\mathcal{K}}

\let\terms\undefined
\newcommand*{\terms}[4]{\mathfrak{T}_{{#1}\ifthenelse{\equal{#2}{}}{}{,{#2}}\ifthenelse{\equal{#3}{}}{}{,{#3}}}({#4})}

\newcommand{\threshold}[1]{\theta\ifthenelse{\equal{#1}{}}{}{({#1})}}

\newcommand{\pseudo}[3]{\mathcal{E}_{{#1},{#2}}\ifthenelse{\equal{#3}{}}{}{[{#3}]}}

\newcommand{\copyof}[1]{\mathsf{copy}_{#1}}
\newcommand{\domof}{\mathit{dom}}
\newcommand{\univof}[1]{\mathit{univ}_{#1}}
\newcommand{\parentof}{\overline{\mathit{parent}}}
\newcommand{\bagof}{\mathsf{bag}}

\newif\ifLongVersion\LongVersiontrue

\ifLongVersion 
\usepackage[createShortEnv, conf={normal}]{proof-at-the-end}
\newenvironment{proofTriangleTextEnd}{}{}
\newenvironment{proofFanTextEnd}{}{}
\else 
\usepackage[createShortEnv,conf={end,no link to proof,restate}]{proof-at-the-end}

\fi

\title{Robust Algebraic Theories of Triangle Graphs}
\author{
  Marius Bozga\inst{1}
  \and
  Radu Iosif\inst{1}
  \and
  Florian Zuleger\inst{2}
}

\institute{
  CNRS, Université Grenoble Alpes, Grenoble, France\\
  \email{Marius.Bozga@univ-grenoble-alpes.fr, Radu.Iosif@univ-grenoble-alpes.fr}
  \and
  Technische Universtit\"{a}t M\"{u}nchen, M\"{u}nchen, Germany\\
  \email{Florian.Zuleger@mytum.de}
}

\authorrunning{Bozga, Iosif, Zuleger}
\titlerunning{Triangle Graphs}

\begin{document}

\maketitle


\begin{abstract}
Triangle graphs are graphs of tree-width at most three in which every
edge belongs to a triangle. This class encompasses well-known graph
families such as Apollonian networks. We also consider fan graphs, a
subclass of triangle graphs closely related to the 3-connected
triangle graphs.

Our main result is an algebraic characterization of both classes. We
introduce two graph algebras based on parallel composition and a
ternary serial composition, and show that they generate exactly the
triangle and fan graphs, respectively. These algebras provide a
natural extension of the classical algebra of series-parallel graphs
from tree-width two to tree-width three.

Building on these characterizations, we investigate context-free,
recognizable, and logically-definable graph languages. We show that
counting monadic second-order logic (CMSO) is decidable over the
context-free sets of triangle and fan graphs. Moreover, we prove that
recognizable graph languages coincide with languages definable in CMSO
for both algebras.
\end{abstract}

\section{Introduction}

Graphs in which every edge belongs to a triangle have been studied
extensively in graph theory and include well-known families such as
Apollonian networks\footnote{The name is derived from the packing of
mutually-tangent circles due to Apollonius of Perga
(262-190BC).}~\cite{PhysRevLett.94.018702,MA2024115486}, which are
obtained by recursively subdividing triangular faces, as in
Figure~\ref{fig:triangles} (a). Such graphs and related graph classes
arise in a variety of settings, including graph drawing, network
design and graph coloring~\cite{BENANTAR199585,GoldnerHarary}. The
motivation of studying these classes of graphs lies in their
usefulness as models of, e.g., evolutionary
processes~\cite{MA2024115486} or resource distribution in parallel
computing~\cite{BENANTAR199585}. Since the term \emph{triangle graph}
is used with different meanings in the literature, we use it here to
denote a simple graph of tree-width at most three in which every edge
belongs to a triangle. We also consider \emph{fan graphs}, a subclass
of triangle graphs that is defined by a connectivity
condition. Triangle and fan graphs are the two classes studied in this
paper.

Our interest in triangle graphs stems from formal language theory: we
are interested in finding finite representations of infinite sets of
graphs and in understanding the expressive power of the resulting
specification formalisms. Such representations play a central role in
many areas of computer science, including verification, synthesis, and
learning. Broadly speaking, one can distinguish constructive
representations, such as grammars and algebras, from descriptive
representations based on logic. A central objective is to understand
the relationships between the corresponding language classes, such as
context-free, recognizable, and logically-definable graph languages.
For instance, the equivalences between recognizability and
definability in Monadic Second Order (\mso) logic for
words~\cite{Buechi90} and trees~\cite{Doner70} are among the pillars
of formal language theory, and constitute the theoretical basis of
automata-based model-checking~\cite{DBLP:books/daglib/0020348} and
synthesis~\cite{Buchi1990} techniques.  Extending such correspondences
to graph languages is a longstanding challenge.  In this paper, we
investigate these questions for the triangle and fan graphs introduced
above.

Among graph classes, graphs of bounded tree-width occupy a
distinguished position.  Bounded tree-width is closely linked to
algorithmic tractability: many \np-complete graph problems such as
Hamiltonicity and $k$-Colorability become \ptime, when restricted to
inputs whose tree-width is bounded, see, e.g., \cite[Chapter
  11]{DBLP:series/txtcs/FlumG06}. Moreover, bounding the tree-width
sets a sharp frontier between the decidability and undecidability of
Monadic Second Order (\mso) logical
theories~\cite{CourcelleI,Seese91}.  Further, the tree-width parameter
is related to connectivity: graphs of tree-width at most $k$ are at
most $k$-connected\footnote{A graph is $k$-\emph{connected} if it has
at least $k$ vertices and cannot be disconnected by deleting fever
than $k$ vertices.}. Connectivity-based graph decompositions (i.e.,
how can $k$-connected graphs be decomposed into $(k-1)$-connected
graphs, etc.) are important tools for understanding the structure of
graphs and devising graph algebras
\cite{Cunningham_Edmonds_1980,10353159}. Unfortunately, little is
currently known about such decompositions, for
$k\geq3$~\cite{DBLP:conf/soda/KurkofkaP26}.

The extension of formal language-theoretic methods from words and
trees to graphs has been particularly successful for series-parallel
graphs, that is, the 2-connected graphs of tree-width at most two.
These graphs admit elegant algebraic descriptions based on serial and
parallel composition, which have led to logical
characterizations~\cite{CourcelleV}, regular
expressions~\cite{DBLP:conf/icalp/Doumane22}, and regular
grammars~\cite{Lics25}.  However, beyond tree-width two, natural graph
classes admitting comparably simple algebraic descriptions remain
scarce.  This raises the question whether the successful theory of
series-parallel graphs can be extended to graph classes of tree-width
three.

Our starting point is the observation that triangle and fan graphs
provide natural candidates for such an extension. We introduce two
graph algebras based on parallel composition and a ternary serial
composition and show that they generate exactly the classes of
triangle and fan graphs, respectively. Thus graph classes defined by
tree-width at most three, the triangle property, and possibly
connectivity constraints admit precise algebraic characterizations. In
this sense, the resulting algebras can be viewed as natural
generalizations of the classical algebra of series-parallel graphs
from tree-width two to tree-width three. For instance, Figure
\ref{fig:triangles} (c) shows a series-parallel graph encoded by a
triangle graph with a central vertex linked by a dashed edge to every
other vertex.

\begin{figure}[t!]
  \centerline{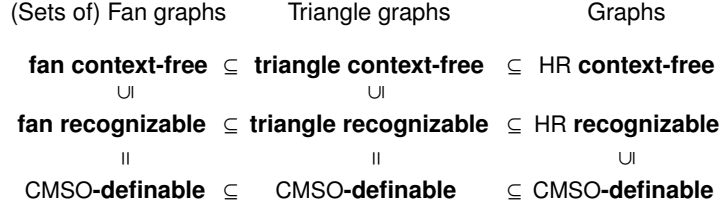}
  \caption{Relations between recognizable, context-free and \cmso-definable sets of graphs}
  \label{fig:sets}
\end{figure}

\vspace*{-\baselineskip}
\paragraph*{Contributions}
We summarize below the contributions of this
paper: \begin{compactenum}
\item We identify the classes of \emph{triangle} graphs, see
  Figure~\ref{fig:triangles} (a), and \emph{fan} graphs, see
  Figure~\ref{fig:triangles} (b) as natural graph-theoretic
  counterparts of series-parallel graphs in the setting of tree-width
  three. The distinction between the two classes is closely tied to
  connectivity: while triangle graphs may have connectivity between
  zero and three, fan graphs are precisely those 2-connected graphs
  that become 3-connected after adding a basic triangle.
\item We define two graph algebras based on the standard
  Courcelle-style parallel composition and a generalized ternary
  serial composition. The first algebra generates the class of
  triangle graphs whereas the second generates the subclass of fan
  graphs.
\item We investigate the structure of set for triangle and fan graphs
  from three complementary perspectives: grammars (context-free
  languages), finite algebras (recognizable languages), and logic
  (CMSO-definable languages). We establish several connections between
  these notions and, in particular, prove that recognizable and
  CMSO-definable languages coincide for both
  algebras. Figure~\ref{fig:sets} summarizes the resulting
  relationships between the language classes considered here.
\end{compactenum}

Beyond their role in describing graph languages, the algebras
introduced in this paper may provide useful tools for studying the
structure of graphs of tree-width at most three.  Possible
applications include the axiomatisation of syntactic equivalence of
graph
expressions~\cite{DBLP:conf/mfcs/DoczkalP18,conf/icalp/Doumane0P24},
and the development of canonical graph decompositions.  A particularly
intriguing connection arises from a recent canonical decomposition
theorem for 3-connected graphs~\cite{10353159}, based on
totally-nested mixed-separators of cardinality
three\footnote{Mixed-separators consist of both vertices and
edges. Totally-nested separators are such that none ``cuts through''
another.}.  The resulting decomposition consists of wheels and
thickened $K_{3,m}$'s (i.e., bipartite graphs with a side being a
$K_3$), both of which admit simple descriptions in our triangle
algebras. This observation suggests that triangle graphs may play a
broader role in the structural theory of graphs of tree-width at most
three and raises the question of whether our algebraic descriptions
can contribute to future decomposition results.

\vspace*{-\baselineskip}
\paragraph*{Related work}
The field of graph algebras was pioneered by Courcelle and Engelfriet
(see~\cite{courcelle_engelfriet_2012} for a comprehensive survey),
with the introduction of hyperedge-replacement (\hr) and
vertex-replacement (\vr) algebras. These algebras consist of low-level
operations that capture the graph-theoretic notions of tree-width (for
\hr) and clique-width (for \vr) as the least number of vertex labels
(also called sources, or ports) used by a term that evaluates to a
given graph. The main inconvenience of these low-level algebras is the
unrestricted use of forget and rename operations, that make it fairly
hard to parse a graph, i.e., retrieve a term that evaluates to that
graph, directly from its structure. This problem has been noticed early
by Arnborg et al.~\cite{10.1007/BFb0017382}, who introduce a
generalized $n$-ary series composition, for $n\geq2$, and give a
linear-time parsing algorithm in this derived algebra.

For graphs of tree-width $2$ at most, an algebra based on generalized
serial composition, that enables the development of regular grammars
capturing its recognizable sets\footnote{Equivalently, definable in
\mso{} with modulo constraints on cardinality of sets.}, has been
introduced in~\cite{Lics25}. This work complements an earlier result
of Doumane, that defines \hr-based regular expressions for the same
class of recognizable graph
languages~\cite{DBLP:conf/icalp/Doumane22}. The generalized serial
composition is also used in~\cite{conf/icalp/Doumane0P24} to give a
finite axiomatisation of term equivalence for the class of graphs of
tree-width at most three. Both our algebras may be viewed as a
tree-width three (at most) analogue of the classical algebra of
series-parallel graphs~\cite{DUFFIN1965303}. These algebras combine
parallel composition with a serial composition, the latter becoming
ternary in the tree-width three setting. In particular, the Apollonian
networks (Figure~\ref{fig:triangles}(a)) are precisely the graphs
generated from basic triangles using ternary serial composition alone.

\vspace*{-\baselineskip}
\paragraph*{Notations}
The set of natural numbers is denoted by $\nat$. The powerset of a
finite set $A$ is denoted $\pow{A}$ and its cardinality by
$\cardof{A}$. The \emph{disjoint union} of two sets $A$ and $B$ is
denoted $A \uplus B$, assumed to be undefined if $A$ and $B$ are not
disjoint. A partial function is denoted $f : A \rightharpoonup B$ and
its domain as $\dom{f}$. The \emph{domain-restriction} of $f$ to a
subset $X \subseteq A$ is the partial function $\proj{f}{X}$ that
agrees with $f$ over $X$ and is undefined outside $X$.

\section{Definitions}

We consider simple undirected binary graphs having some vertices
designated by labels from $\nat$. Formally, a graph is a triple
$\graph=(\vertof{},\edgeof{},\sourceof{})$, where $\vertof{}$ is a
finite set of vertices, $\edgeof{} \subseteq \set{e \in
  \pow{\vertof{}} \mid \cardof{e} = 2}$ is a set of undirected edges
(i.e., there are no duplicated edges or self-loops) and $\sourceof{} :
\nat \rightharpoonup \vertof{}$ is a partial injective function,
having finite domain, that labels vertices with numbers. The
\emph{type} of $\graph$ is the finite set $\dom{\sourceof{}} \subseteq
\nat$. We denote by $\vertof{\graph}$, $\edgeof{\graph}$ and
$\sourceof{\graph}$ the components of $\graph$, respectively.

The \emph{sources} of $\graph$ are the vertices in the range of
$\sourceof{\graph}$, omitted when they are not important. The vertices
that are not sources are called \emph{inner vertices}. A
\emph{$n$-graph} is a graph of type $\set{1, 2, \ldots, n}$. An
\emph{empty $n$-graph} is a $n$-graph having no other vertices than
its sources and no edges. An \emph{unlabeled} graph is a $0$-graph,
i.e., a graph without sources.

A set of vertices $X \subseteq \vertof{\graph}$ of a graph $\graph$
\emph{induces} the subgraph $\graph[X] \isdef
(X,\edgeof{\graph}\cap\pow{X},\proj{\sourceof{\graph}}{X})$ of
$\graph$. We denote by $\graph - X$ the graph obtained from $\graph$
by deleting all vertices in $X$ and all edges incident to a vertex in
$X$, i.e., $\graph - X \isdef \graph[\vertof{\graph} \setminus X]$. A
graph $\graph$ is said to be $(A,B)$-\emph{bipartite} if
$\vertof{\graph}=A \uplus B$ and each edge is between a vertex in $A$
and one in $B$.

In the following, we shall not distinguish between isomorphic graphs,
i.e., graphs that differ only by a renaming of vertices. An edge
$\set{x,y}\in\edgeof{\graph}$ will be denoted as $xy$. A \emph{bridge}
is a $2$-graph consisting of two vertices and an edge between them. We
denote by $xyz$ the unlabeled graph consisting of the vertices $x,y,z$
and edges $xy$, $yz$ and $xz$, and by $K_3$ any graph whose underlying
unlabeled graph is of this form.

A path in a graph $\graph$ is a sequence of vertices $x_1,\ldots,x_n
\in \vertof{\graph}$, for some $n\geq2$, such that
$x_1,\ldots,x_{n-1}$ are pairwise distinct and $x_ix_{i+1} \in
\edgeof{\graph}$, for each $1 \leq i < n$. A graph is \emph{connected}
if between any two vertices there is a path, otherwise it is
\emph{disconnected}. A graph is $k$-\emph{connected}, for $k\geq1$, if
it has at least $k+1$ vertices and deleting fewer than $k$ vertices
and their incident edges does not disconnect it. For instance, $K_3$
is $2$-connected, but not $3$-connected.

A graph is \emph{acyclic} if it contains no path $x_1, \ldots, x_n$
such that $n\geq3$ and $x_1=x_n$. A \emph{tree} is a connected and
acyclic graph. The vertices of a tree are called \emph{nodes}. Note
that we consider non-rooted trees, i.e., trees without a designated
root node, and chose to introduce rooted trees later, when needed.

\begin{definition}\label{def:tw}
  A \emph{tree decomposition} of a graph $\graph$ is a pair
  $(\tree,\bag)$, where $\tree$ is a tree and $\bag : \vertof{\tree}
  \rightarrow \pow{\vertof{\graph}}$ is a labeling of the nodes in
  $\tree$ with sets of vertices from $\graph$, called \emph{bags},
  such that: \begin{compactenum}
  \item\label{it1:def:tw} the sources of $\graph$ are all contained in some bag,
  \item\label{it2:def:tw} each edge of $\graph$ is a subset of some bag,
  \item\label{it3:def:tw} for each vertex $x \in \vertof{\graph}$, the set of nodes of
    $\tree$ whose bags contain $x$ induce a nonempty and connected
    subgraph of $\tree$.
  \end{compactenum}
  The width of $(\tree,\bag)$ is $\width{\tree,\bag} \isdef
  \max\{\cardof{\bag(n)} \mid n \in \vertof{\tree}\}-1$ and the
  \emph{tree-width} of $\graph$ is $\twd{\graph} \isdef \min
  \{\width{\tree,\bag} \mid (\tree,\bag) \text{ is a
    tree-decomposition of } \graph\}$. A tree decomposition
  $(\tree,\bag)$ of a graph $\graph$ is \emph{optimal} if
  $\width{\tree,\bag} = \twof{\graph}$.
\end{definition}

For a node $n \in \vertof{\tree}$ and an edge $nm\in\edgeof{\tree}$,
the set of vertices $\bag(n)\cap\bag(m)$ is called the \emph{adhesion}
of $n$ for $m$, denoted $\adh{\tree}{n}{m}$, or $\adh{}{n}{m}$,
whenever the tree decomposition in question is understood. In the
following, we assume that each tree decomposition has the following property:

\begin{assumption}\label{ass:td}
  Let $(\tree,\bag)$ be a tree decomposition. For each edge $nm \in
  \edgeof{\tree}$, we have $\bag(n)\setminus\adh{}{n}{m} \neq
  \emptyset$.
\end{assumption}
This assumption loses no generality, because each edge $nm$ of $\tree$
such that $\bag(n)\subseteq\bag(m)$ can be contracted, the result
being a tree decomposition of the same graph, having the same width.

\subsection{Triangle Graphs}

\begin{figure*}[t!]
  \centerline{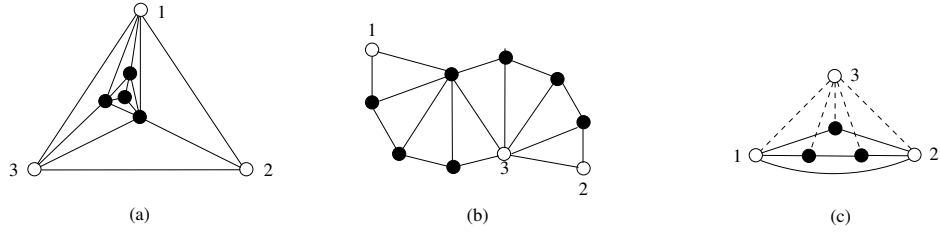}
  \caption{Triangle graph (a) Fan graph (b) Series-parallel graph
    encoded by triangle graph (c). Hollow circles denote sources and
    filled circles are inner vertices.}
  \label{fig:triangles}
\end{figure*}

First, we define triangle graphs by means of a so-called
\emph{triangle property}, which requires that each edge be part of a
$K_3$ induced subgraph:

\begin{definition}\label{def:triangle-property}
  A graph $\graph$ has the \emph{triangle property} if and only if,
  for each edge $xy\in\edgeof{\graph}$ there exists a vertex $z \in
  \vertof{\graph}\setminus\set{x,y}$ and edges $xz, yz \in
  \edgeof{\graph}$. The vertex $z$ is said to be a \emph{witness} of
  the edge $xy$ in $\graph$.

  Given a graph $\graph \in \trianglegraphs$ and a set $X \subseteq
  \vertof{\graph}$ the \emph{triangle subgraph of $\graph$ induced by
  $X$} is the graph $\graph[X]^\basictriangle \isdef
  (X,\edgeof{},\sourceof{})$, where
  $\edgeof{}\subseteq\edgeof{\graph}$ is the set of edges whose
  endpoints and witness all belong to $X$ and $\sourceof{}$ is the
  domain-restriction of $\sourceof{\graph}$ to $X$, i.e.,
  $\sourceof{}\isdef\proj{\sourceof{\graph}}{X}$.
\end{definition}
For instance, any clique (i.e., a graph having an edge between any two
vertices) has the triangle property. The following definition makes
the notion of triangle graphs precise:

\begin{definition}\label{def:triangle-graph}
  A \emph{triangle graph} is a $3$-graph of tree-width at most $3$
  having the triangle property. The set of triangle graphs is denoted
  as $\trianglegraphs$.
\end{definition}
For instance, a clique with more than $4$ vertices has tree-width more
than $3$, hence is not a triangle graph. On the other hand, the graphs
in Figure \ref{fig:triangles} are triangle graphs because they have
tree-width $3$ and the triangle property. As a remark, any $3$-graph
must have tree-width at least $2$, because the three sources of the
graph must belong to some bag, by Definition \ref{def:tw}
(\ref{it1:def:tw}). Hence, triangle graphs have tree-width between $2$
and $3$. Note that there is no restriction on the connectivity of a
triangle graph, that can be disconnected, $1$-, $2$- or $3$-connected.

Second, we define a subclass of triangle graphs, that are almost
$3$-connected. A \emph{thickening} of a graph $\graph$ is a graph
$\overline{\graph}$ obtained by adding an edge between each pair of
sources.

\begin{definition}\label{def:fan-graph}
  A \emph{fan graph} is a $2$-connected triangle graph $\graph$ such
  that exactly one of the following holds: \begin{compactenum}
  \item $\graph$ is $K_3$, 
  \item $\graph$ is $3$-connected,
  \item $\overline\graph$ is $3$-connected.
  \end{compactenum}
  The set of fan graphs is denoted as $\trianglefangraphs$.
\end{definition}
\noindent For instance, the graphs in Figure \ref{fig:triangles} are
fan graphs.  On the other hand, the $3$-graph consisting of four
disconnected vertices is a triangle graph, but not a fan graph.

\section{Algebras of Triangle Graphs}
\label{sec:algebras}

We characterize triangle graphs by means of graph algebras, i.e., sets
of graphs that are closed under certain operations. To begin with, we
recall a standard graph algebra, known as \emph{hyperedge replacement}
(\hr) ~\cite{CourcelleI} and define the algebra of triangle graphs as
a derived algebra of \hr. The elements of the \hr{} algebra are graphs
of any finite type $\tau \subseteq \nat$ and its operations are as
follows:

\begin{figure*}[t!]
  \centerline{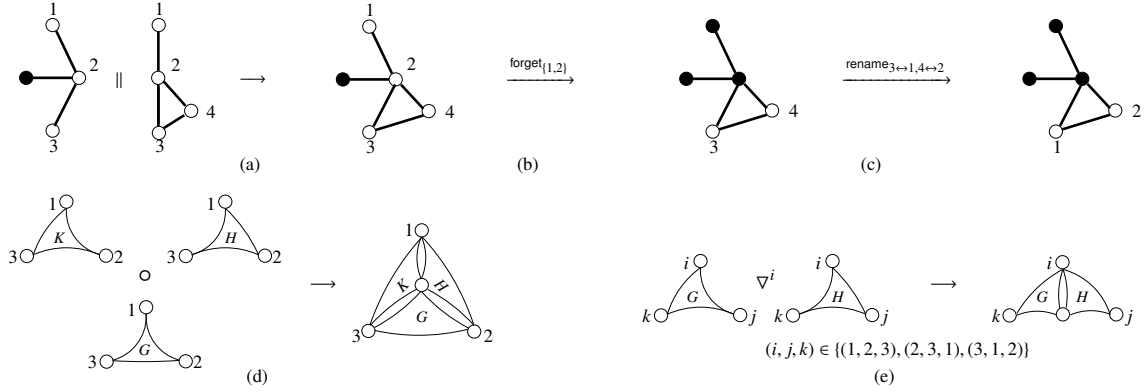}
  \caption{The parallel composition (a) forget (b) and rename (c)
    operations of the \hr{} algebra. Serial composition (d) and fan
    (e) operations of the triangle algebras.}
  \label{fig:graphs}
\end{figure*}

\begin{compactitem}[-]
\item $\vertex$ is a constant denoting an empty $1$-graph,
\item $\bridge$ is a constant denoting a bridge, i.e., a $2$-graph
  consisting of one edge and its endpoints,
\item $\forget{i}(\graph)$ removes the label of $\sourceof{\graph}(i)$
  (not the vertex itself); for a finite set $I \subseteq \nat$, we
  denote by $\forget{I}(\graph)$ the composition of the functions
  $\forget{i}$, for each $i \in I$, where the order is not important,
\item $\rename{i\leftrightarrow j}(\graph)$ swaps the labels of
  $\sourceof{\graph}(i)$ and $\sourceof{\graph}(j)$, for some $i \neq
  j \in \nat$; we denote by $\rename{i_1\leftrightarrow j_1, \ldots,
    i_k\leftrightarrow j_k}$ the composition of the functions
  $\rename{i_1\leftrightarrow j_1} \ldots \rename{i_k\leftrightarrow j_k}$
  given in this order,
\item $\graph_1 \pop \graph_2$ (\emph{parallel composition}) joins the
  sources having the same label in the disjoint union (i.e., the union
  of two disjoint copies) of $\graph_1$ and $\graph_2$ and deletes
  duplicated edges.
\end{compactitem}
Figure \ref{fig:graphs} gives examples of applying the parallel
composition (a) forget (b) and rename (c) operations. For simplicity,
in the following we use the same notation for a function symbol and
its interpretation, whenever the algebra in question is clear from the
context. Note that we consider \hr{} operations to be polymorphic,
because they apply to graphs of any type. For instance, two graphs of
type $\set{1,2,3}$ and $\set{1,2,3,4}$ are composed in parallel in
Figure \ref{fig:graphs} (a), the result being a graph of type
$\set{1,2,3,4}$, to which $\forget{\set{1,2}}$ is applied, resulting
in a graph of type $\set{3,4}$, in Figure \ref{fig:graphs} (b). Then,
$\rename{3\leftrightarrow1,4\leftrightarrow2}$ is applied to this
graph to obtain a graph of type $\set{1,2}$.

The signature of the \hr{} algebra is the set $\hrsignature \isdef
\set{\vertex,\bridge,\pop} \cup \set{\forget{i} \mid i \in \nat} \cup
\set{\rename{i\leftrightarrow j} \mid i\neq j \in \nat}$. A term is
built as usual from the function symbols in $\hrsignature$ and
variables, i.e., symbols of zero arity, not in $\hrsignature$. We
denote by $t[x_1,\ldots,x_n]$ a term whose variables are
$x_1,\ldots,x_n$, where a variable may occur one or more times. A term
is \emph{ground} if it contains no variables. The graph obtained by
interpreting the function symbols from the ground term $t$ in the
\hr{} algebra is denoted $\hval(t)$, also called the \emph{value} of
$t$. For instance, the leftmost graph in Figure \ref{fig:graphs} (a)
is $\hval(\bridge \pop \forget{1}(\bridge) \pop \rename{1
  \leftrightarrow 3}(\bridge))$. Although, in general, there might be
several ground terms that evaluate to the same graph, each graph is
the value of at least one ground term, see e.g.~\cite[Proposition
  2.33]{courcelle_engelfriet_2012} for a proof of this statement. The
\hr{} algebra allows to give an alternative equivalent definition of
tree-width (Definition \ref{def:tw}):

\begin{lemma}[Proposition 1.19 in \cite{courcelle_engelfriet_2012}]
  \label{lemma:hr-tw}
  A graph $\graph$ has tree-width at most $k$ if and only if there
  exists a ground term $t$ such that $\hval(t)=\graph$ and $t$ uses at
  most $k+1$ source labels.
\end{lemma}

\vspace*{-\baselineskip}
\paragraph*{Triangle algebra.}
The signature of the \emph{triangle algebra} $\trianglealg$ is the set
$\trianglesign\isdef\set{\basictriangle,\emptygraph,\sop,\pop}$, where
the parallel composition is interpreted the same as in \hr{} and the
other operations are defined below: \begin{compactitem}[-]
\item $\basictriangle$ is the constant symbol denoting the $K_3$
  $3$-graph, formally:
  \[\basictriangle \isdef\ \bridge \pop \rename{1 \leftrightarrow
    3}(\bridge) \pop \rename{2 \leftrightarrow 3}(\bridge)\]
\item $\emptygraph$ is the constant symbol denoting the empty
  $3$-graph, formally:
  \[\emptygraph\ \isdef \vertex \pop \rename{1 \leftrightarrow
    2}(\vertex) \pop \rename{1 \leftrightarrow 3}(\vertex)\]
\item $\sop(\graph,\hgraph,\kgraph)$ (\emph{serial composition}) is
  the operation depicted in Figure \ref{fig:graphs} (d), formally:
  \[\sop(\graph,\hgraph,\kgraph) \isdef \forget{4}(\rename{1
    \leftrightarrow 4}(\graph) \pop \rename{2 \leftrightarrow
    4}(\hgraph) \pop \rename{3 \leftrightarrow 4}(\kgraph))\]
\end{compactitem}
Note that the occurrences of the \hr{} forget operation are confined
to the serial composition. A \emph{triangle term} is a term built from
the function symbols in $\trianglesign$ and variables. The value
$\hval(t)$ of a ground triangle term $t$ is the value of the \hr{}
term obtained by replacing the symbols from $\trianglesign$ with their
definitions.

\vspace*{-\baselineskip}
\paragraph*{Fan algebra.}
The signature of the \emph{fan algebra} $\fanalg$ is the set
$\trianglefansign \isdef
\set{\basictriangle,\fan{1},\fan{2},\fan{3},\sop,\pop}$, whose
operations, depicted in Figure \ref{fig:graphs} (e), have the
following formal definitions, where $\basictriangle$, $\sop$ and $\pop$ are the same as above:
\begin{align*}
  \graph \fan{1} \hgraph \isdef & ~\sop(\emptygraph,\graph,\hgraph) \hspace*{5mm}
  \graph \fan{2} \hgraph \isdef \sop(\graph,\emptygraph,\hgraph) \hspace*{5mm}
  \graph \fan{3} \hgraph \isdef \sop(\graph,\hgraph,\emptygraph)
\end{align*}
Note that $\trianglefansign$ is obtained from $\trianglesign$ by
restricting the empty graph symbol $\emptygraph$ to occur at most once
within a serial composition. A \emph{fan term} is a term over the
signature $\trianglefansign$. The value $\hval(t)$ of a ground fan
term $t$ is the value of the triangle term obtained by replacing the
fan operations from $t$ with their definitions.

\subsection{Relation with Series-Parallel Graphs}

The algebra of series-parallel graphs consists of $2$-graphs built
from bridges ($\bridge$) using binary serial composition (i.e.,
$\graph_1 \sop \graph_2$ joins the $2$-source of $\graph_1$ with the
$1$-source of $\graph_2$, the resulting vertex being unlabeled) and
parallel composition~\cite{DUFFIN1965303}. We denote by $\spgraphs$
the set of series-parallel graphs. We prove that these are exactly the
graphs obtained from basic triangles composed using parallel
composition and the $\fan{3}$ operation, from which the $3$-source is
deleted, together with all the edges incident to it, as in Figure
\ref{fig:triangles} (c).

Formally, given a $2$-graph $\graph$, we denote by
$\widetilde{\graph}$ the $3$-graph obtained from $\graph$ by adding a
$3$-source and connect it via an edge to each vertex from
$\vertof{\graph}$. The set of ground terms consisting of
$\basictriangle$, $\fan{3}$ and $\pop$ is denoted
$\mathcal{T}_3$. Clearly, the set of interpretations of the terms in
$\mathcal{T}_3$ forms a subset of $\trianglefangraphs$. The following
statement formalizes the relation between series-parallel and fan
graphs:

\begin{propositionE}
  $\set{\widetilde{\graph} \mid \graph \in \spgraphs} = \set{\hval(t)
    \mid t \in \mathcal{T}_3}$. 
\end{propositionE}
\begin{proofE}
  ``$\subseteq$'' Let $\graph \in \spgraphs$ be a series-parallel
  graph. There exist a ground term over the operations $\bridge$,
  $\sop$ and $\pop$ such that $\hval(t)=\graph$. Let $\widetilde{t}$
  be the ground term obtained from $t$ by substituting each occurrence
  of $\bridge$ by $\basictriangle$ and of $\sop$ by
  $\fan{3}$. Clearly, we have $\widetilde{t} \in \mathcal{T}_3$. We
  prove that $\widetilde{\graph} = \hval(\widetilde{t})$ by induction
  on the structure of $t$, by considering the cases
  below: \begin{compactitem}[-]
  \item $t = \bridge$: in this case, $\widetilde{t} =
    \basictriangle$ and $\widetilde{\hval(t)} = \hval(\widetilde{t})$
    is obvious.
  \item $t = t_1 \sop t_2$: in this case, $\widetilde{t} =
    \widetilde{t}_1 \fan{3} \widetilde{t}_2$. Note that $\fan{3}$
    joins the $2$-source of $\hval(t_1)$ with the $1$-source of
    $\hval(t_2)$, removes the source label of the joined node, and
    joins the $3$-sources of both $\hval(t_1)$ and $\hval(t_2)$. We
    compute: \begin{align*}
      \widetilde{\hval(t)} = & ~\widetilde{\hval(t_1) \sop \hval(t_2)} \\
      = & ~\widetilde{\hval(t_1)} \fan{3} \widetilde{\hval(t_2)} \text{, by the above argument} \\
      = & ~\hval(\widetilde{t}_1) \fan{3} \hval(\widetilde{t}_2) \text{, by the inductive hypothesis} \\
      = & ~\hval(\widetilde{t})
    \end{align*}
  \item $t = t_1 \pop t_2$: in this case, $\widetilde{t} =
    \widetilde{t}_1 \pop \widetilde{t}_2$. Note that $\pop$ joins the
    $i$-the sources of $\hval(t_1)$ and $\hval(t_2)$, for
    $i=1,2,3$. We compute: \begin{align*}
      \widetilde{\hval(t)} = & ~\widetilde{\hval(t_1) \pop \hval(t_2)} \\
      = &~\widetilde{\hval(t_1)} \pop \widetilde{\hval(t_2)} \text{, by the above argument} \\
      = & ~\hval(\widetilde{t}_1) \pop \hval(\widetilde{t}_2) \text{, by the inductive hypothesis} \\
      = & ~\hval(\widetilde{t})
    \end{align*}
  \end{compactitem}

  \noindent''$\supseteq$'' Let $\widetilde{t} \in \mathcal{T}_3$ be a
  ground term and $t$ be the ground term obtained from $t$ by
  replacing each occurrence of $\basictriangle$ with $\bridge$ and of
  $\fan{3}$ with $\sop$. Clearly, $t$ is a ground term over the
  signature of series-parallel graphs. We prove that
  $\hval(\widetilde{t}) = \widetilde{\hval(t)}$ by induction on the
  structure of $\widetilde{t}$, considering the cases
  below: \begin{compactitem}[-]
  \item $\widetilde{t} = \basictriangle$: in this case $t = \bridge$
    and $ \hval(\widetilde{t}) = \widetilde{\hval(t)}$ is obvious. 
  \item $\widetilde{t} = \widetilde{t}_1 \fan{3} \widetilde{t}_2$: in
    this case, $t = t_1 \sop t_2$. We compute:
    \begin{align*}
      \hval(\widetilde{t}) = &~ \hval(\widetilde{t}_1) \fan{3} \hval(\widetilde{t}_2) \\
      = &~ \widetilde{\hval(t_1)} \fan{3} \widetilde{\hval(t_2)} \text{, by the inductive hypothesis} \\
      = &~ \widetilde{\hval(t_1) \sop \hval(t_2)} \\
      = &~ \widetilde{\hval(t_1 \sop t_2)} 
    \end{align*}
    The second step uses the same argument as in the case $t=t_1 \sop
    t_2$ of ``$\subseteq$''. 
  \item $\widetilde{t} = \widetilde{t}_1 \pop \widetilde{t}_2$: in
    this case, $t = t_1 \pop t_2$. We compute:
    \begin{align*}
      \hval(\widetilde{t}) = &~ \hval(\widetilde{t}_1) \pop \hval(\widetilde{t}_2) \\
      = &~ \widetilde{\hval(t_1)} \pop \widetilde{\hval(t_2)} \text{, by the inductive hypothesis} \\
      = &~ \widetilde{\hval(t_1) \pop \hval(t_2)} \\
      = &~ \widetilde{\hval(t_1 \pop t_2)} 
    \end{align*}
    The second step uses the same argument as in the case $t=t_1 \pop
    t_2$ of ``$\subseteq$''. 
  \end{compactitem}
\end{proofE}

\subsection{Algebraic Characterization of Triangle Graphs}

In the remainder of this section, we prove that the triangle and fan
algebras correctly define the sets of triangle and fan graphs
(Definitions \ref{def:triangle-graph} and \ref{def:fan-graph},
respectively).

\begin{theorem}\label{thm:triangle-graphs}
  $\trianglegraphs = \set{\hval(t) \mid t \text{ is a ground
      triangle term}}$.
\end{theorem}

Since triangle terms translate to \hr{} terms that use at most $4$
sources, the tree-width of a graph $\hval(t)$, where $t$ is a ground
triangle term, is necessarily less than or equal to $3$ (Lemma
\ref{lemma:hr-tw}). The right to left direction of Theorem
\ref{thm:triangle-graphs} is a consequence of this and the following
fact:

\begin{lemmaE}\label{lemma:witness-triangle}
  For each ground triangle term $t$, the graph $\hval(t)$ has the
  triangle property.
\end{lemmaE}
\begin{proofE}
  Let $\graph \isdef \hval(t)$ and $xy\in\edgeof{\graph}$ be an edge
  of $\graph$. We shall prove the existence of a vertex
  $z\in\vertof{\graph}$ such that $xz,yz\in\edgeof{\graph}$.
  The proof goes by induction on the structure of $t$,
  considering the following cases: \begin{itemize}[-]
  \item $t=\emptygraph$: in this case, there is no edge
    $xy\in\edgeof{\graph}$, so the triangle property holds vacuously.
  \item $t=\basictriangle$: in this case, $z$ is the only vertex of
    $\graph$ distinct from $x$ and $y$. Then,
    $xy,yz\in\edgeof{\graph}$ because $\hval(t)$ is a $K_3$ graph.
  \item $t=\sop(t_1,t_2,t_3)$: since one of $xy \in
    \edgeof{\hval(t_1)}$, $xy \in \edgeof{\hval(t_2)}$ or $xy \in
    \edgeof{\hval(t_3)}$ must hold, the inductive hypothesis applies,
    hence $z \in \vertof{\hval(t_i)}$ and $xz,yz \in
    \edgeof{\hval(t_i)}$, where $1 \leq i \leq 3$ is such that $xy \in
    \edgeof{\hval(t_i)}$.
  \item $t=t_1 \pop t_2$: since one of $xy \in \edgeof{\hval(t_1)}$ or
    $xy \in \edgeof{\hval(t_2)}$, the inductive hypothesis applies,
    hence $z \in \vertof{\hval(t_i)}$ and $xz,yz \in
    \edgeof{\hval(t_i)}$, where $1 \leq i \leq 2$ is such that $xy \in
    \edgeof{\hval(t_i)}$.
  \end{itemize}
  \vspace*{-1.5\baselineskip}
\end{proofE}

The left to right direction of Theorem \ref{thm:triangle-graphs}
requires additional definitions and lemmas. A tree $\tree$ is
\emph{rooted} if it has at least two nodes and a designated
\emph{root} $r \in \vertof{\tree}$. The \emph{leaves} of a rooted tree
are the nodes different from the root that belong to exactly one edge
of $\tree$. A node $m \in \vertof{\tree}$ is an \emph{ancestor} of
another node $n \in \vertof{\tree}$ if $m$ is on the unique path from
$r$ to $n$. If, moreover, $mn \in \edgeof{\tree}$, we say that $m$ is
the \emph{parent} of $n$ (resp. $n$ is a \emph{child} of $m$) in
$\tree$.  In the following, we shall work with rooted tree
decompositions of the following form:

\begin{definition}\label{def:alternating-td}
  An \emph{alternating tree decomposition} $(\tree,\bag)$ of a graph
  $\graph$ is a tree decomposition of $\graph$ such that $\tree$ is
  $(A,B)$-bipartite, has root $r\in A$, and the following
  hold: \begin{compactenum}
  \item\label{it1:def:alternating-td} the children of a node $n \in A$
    (resp. $n \in B$) belong to $B$ (resp. $A$) and the leaves of
    $\tree$ belong to $B$,
  \item\label{it2:def:alternating-td} for each node $n \in A
    \setminus\set{r}$ having parent $m$, for each child $p$ of $n$, we
    have $\bag(n) = \bag(m) \cap \bag(p)$,
  \item\label{it3:def:alternating-td} for each node $n \in B$ having
    parent $m$, for each child $p$ of $n$, we have
    $\bag(m)\neq\bag(p)$.
  \item\label{it4:def:alternating-td} $\bag(r)$ contains the sources
    of $\graph$ and only those vertices.
  \end{compactenum}
  The nodes from the sets $A$ and $B$ are called the \emph{adhesion}
  and \emph{bag} nodes of $\tree$, respectively.
\end{definition}
The following proves that considering only alternating tree
decompositions loses no generality:

\begin{lemmaE}\label{lemma:alternating-td}
  For each tree decomposition $(\tree,\bag)$ of a graph $\graph$,
  there exists an alternating tree decomposition $(\tree',\bag')$ of
  $\graph$, such that $\width{\tree',\bag'} \leq \width{\tree,\bag}$.
\end{lemmaE}
\begin{proofE}
  We build $(\tree',\bag')$ from $(\tree,\bag)$ as follows. First, we
  chose a node $r \in \vertof{\tree}$ such that $\bag(r)$ contains the
  sources of $\graph$, to be the root of $\tree$. Such a node exists,
  by Definition \ref{def:tw} (\ref{it1:def:tw}). We proceed top-down
  on the structure of the rooted tree $\tree$, starting from $r$. Let
  $m \in \vertof{\tree}$ be a node and $n_1, \ldots, n_k \in
  \vertof{\tree}$ be the children of $m$ in $\tree$. By a reindexing,
  if necessary, we consider indices $1 = i_1 \leq \ldots \leq i_\ell =
  k$ such that: \begin{itemize}[-]
  \item $\adh{\tree}{n_j}{m}=\adh{\tree}{n_h}{m}$, for all $1 \leq s <
    \ell$ and $i_s \leq j < h < i_{s+1}$, and
  \item $\adh{\tree}{n_{i_j}}{m} \neq \adh{\tree}{n_{i_h}}{m}$, for
    all $1 \leq j < h \leq \ell$.
  \end{itemize}
  That is, we group the children of $m$ according to their adhesions
  in $\tree$. We consider fresh nodes $p_1, \ldots, p_{\ell}$ that
  will become the children of $m$ in $\tree'$. Moreover, for all $1
  \leq s \leq \ell$: \begin{itemize}[-]
  \item $n_{i_s}, \ldots, n_{i_{s+1}-1}$ are the children of $p_s$ in $\tree'$,
  \item $\bag'(p_s) = \adh{\tree}{n_{i_s}}{m}$.
  \end{itemize}
  Otherwise, $\bag'$ agrees with $\bag$ over $\vertof{\tree}$. The
  tree $\tree'$ is the result of applying this transformation
  top-down, starting from $r$ (independent branches can be processed
  in any order). We are left with checking the four points of
  Definition \ref{def:alternating-td} below: \begin{compactitem}[-]
    \item It is easy to check that $\tree'$ is bipartite, with bag
      nodes $\vertof{\tree}$ and adhesion nodes all the fresh nodes
      introduced by the construction, thus taking care of point
      (\ref{it1:def:alternating-td}).
    \item For each adhesion node $p$, we have $\bag'(p) =
      \adh{\tree}{m}{n} = \bag'(m) \cap \bag'(n)$, where $m$ is the
      bag parent of $p$ and $n$ is any of its bag children in $\tree'$
      (i.e., $n$ is a child of $m$ in $\tree$). This takes care of
      point (\ref{it2:def:alternating-td}).
    \item In order to take care of point
      (\ref{it3:def:alternating-td}), we change $\tree'$ as follows:
      for each non-root bag node $n$ with adhesion parent $m$, if $n$
      has an adhesion child $p$ such that $\bag'(m)=\bag'(p)$, we
      append $p$ to the parent of $m$. Moreover, $\width{\tree',\bag'}
      \leq \width{\tree,\bag}$ because each bag of $(\tree',\bag')$ is
      a subset of some bag of $(\tree,\bag)$.
    \item Finally, we add a new adhesion node $r'$, connect it to $r$,
      set $\bag'(r)$ to be the set of sources of $\graph$ and promote
      $r'$ to be new root of $\tree'$, thus taking care of point
      (\ref{it4:def:alternating-td}).
  \end{compactitem}
\end{proofE}

Moreover, we shall consider alternating tree decompositions having the
following property:

\begin{definition}\label{def:step-td}
  An alternating tree decomposition $(\tree,\bag)$, whose sets of sets
  of adhesion and bag nodes are $A$ and $B$, respectively, has the
  $k$-\emph{step} property, for some $k\geq1$, if and only if
  $\cardof{\bag(n)}=k$ if $n \in A$ and $\cardof{\bag(n)}=k+1$ if $n
  \in B$.
\end{definition}

Focusing on alternating tree decompositions having the $k$-step
property does not lose generality, as shown by the following lemma:

\begin{lemmaE}\label{lemma:full-td}
  For each $k$-graph $\graph$ of tree-width $k$, there exists an
  optimal alternating tree decomposition $(\tree,\bag)$ having the
  $k$-step property.
\end{lemmaE}
\begin{proofE}
  Let $(\tree',\bag')$ be a tree decomposition of $\graph$ such that
  $\width{\tree',\bag'}=k$. Then there exists at least one node $n \in
  \vertof{\tree'}$ such that $\cardof{\bag'(n)}=k+1$. We obtain
  $(\tree,\bag)$ by applying the construction from the proof of Lemma
  \ref{lemma:alternating-td} to the tree decomposition
  $(\tree''',\bag''')$ obtained from $(\tree',\bag')$ in two phases:

  \noindent(\emph{Flooding}) We repeat the following steps as long as
  there is a node in $\tree'$ whose bag has cardinality at most
  $k$: \begin{itemize}
  \item let $n$ be a node that has at least one neighbour such that
    $\cardof{\bag'(n)} > \cardof{\bag'(m)}$.
  \item chose a vertex $v \in \bag'(n) \setminus \bag'(m)$ and add it
    to $\bag'(m)$.
  \end{itemize}
  The iteration of the above steps terminates because the total number
  of vertices in the bags of $\tree'$ increases with each iteration
  and this number cannot exceed $\cardof{\vertof{\tree'}} \cdot
  (k+1)$. Moreover, upon termination each node has a bag of
  cardinality $k+1$. The last step of this phase is the contraction of
  all edges $nm \in \edgeof{\tree'}$ such that
  $\bag'(n)=\bag'(m)$. The outcome, denoted $(\tree'',\bag'')$, is a
  tree decomposition of $\graph$ because, for each vertex $v \in
  \vertof{\graph}$, the set of nodes $n \in \vertof{\tree''}$ such
  that $v \in \bag''(n)$ stays connected throughout the
  iteration of the above steps.

  \noindent(\emph{Smoothing}) We repeat the following steps until
  nothing changes: \begin{itemize}
  \item let $nm \in \edgeof{\tree''}$ be an edge such that $\bag''(m)
    \setminus \adh{\tree''}{n}{m} = \set{v_1,\ldots,v_\ell}$ and
    $\bag''(n) \setminus \adh{\tree''}{n}{m} =
    \set{w_1,\ldots,w_\ell}$, for some $2 \leq \ell \leq k$. Note
    that, because $\cardof{\bag''(m)} = \cardof{\bag''(n)} = k+1$, we
    must have $\cardof{\bag''(m) \setminus \adh{\tree''}{n}{m}} =
    \cardof{\bag''(n) \setminus \adh{\tree''}{n}{m}} \leq k$, where
    the sets $\set{v_1,\ldots,v_\ell}$, $\set{w_1,\ldots,w_\ell}$ and
    $\adh{\tree''}{n}{m}$ are pairwise disjoint.
  \item let $n_1, \ldots, n_{\ell-1}$ be fresh nodes with bags
    $\bag'''(n_i) \isdef \set{v_1,\ldots,v_{\ell-i}} \uplus
    \set{w_1,\ldots,w_i} \uplus \adh{\tree''}{n}{m}$, for all $1 \leq
    i \leq \ell-1$. It is easy to check that $\cardof{\bag'''(n_i)} = k +
    1$, for all $1 \leq i \leq \ell-1$.
  \item remove the edge $nm$ and add edges $nn_1, n_1n_2, \ldots,
    n_{\ell-1}m$.
  \end{itemize}
  Let $(\tree''',\bag''')$ be the outcome of this iteration, where
  $\bag'''$ agrees with $\bag''$ over $\vertof{\tree''}$. Note that
  $(\tree''',\bag''')$ is a tree decomposition of $\graph$ because,
  for each vertex $v \in \vertof{\graph}$, the set of nodes $n \in
  \vertof{\tree'''}$ such that $v \in \bag'''(n)$ stays connected
  throughout the above iteration.

  To summarize, $(\tree''',\bag''')$ is a tree decomposition of
  $\graph$ having the following properties: \begin{enumerate}
  \item $\cardof{\bag'''(n)} = k+1$, for all nodes $n \in
    \vertof{\tree'''}$,
  \item $\cardof{\bag'''(n) \setminus \adh{\tree'''}{n}{m}} = 1$, for
    all edges $nm \in \edgeof{\tree'''}$.
  \end{enumerate}
  Then, applying the construction from the proof of Lemma
  \ref{lemma:alternating-td}, yields an alternating tree decomposition
  $(\tree,\bag)$ with the required properties. In particular, the bag
  of the root of $(\tree,\bag)$ has $k$ vertices because $\graph$ has
  $k$ sources.
\end{proofE}
\noindent Among the consequences of the above lemma is that $k$-graphs
of tree-width at most $k$ are at most $k$-connected. This occurs
because each adhesion of a tree decomposition of a graph is a
separator of that graph (i.e., splits the graph in two sides such that
each path between vertices belonging to different sides must contain a
vertex of the separator). See Lemma \ref{lemma:adhesion-separation}
below for a proof of this statement.

The following notion will be used in the proofs of the
characterization results for triangle and fan graphs. Given an optimal
tree decomposition of a graph $\graph$ of tree-width $k$ at most, a
$k$-\emph{history} associates each vertex of the graph a number
between $1$ and $k+1$. Intuitively, this is the source number assigned
to the vertex during the construction of the graph, based on the given
tree decomposition. More precisely, any optimal tree decomposition of
$\graph$ can be translated into a ground \hr{} term $t$ such that
$\graph=\hval(t)$ and $t$ uses at most $k+1$ source labels (Lemma
\ref{lemma:hr-tw}). The $k$-history corresponding to $t$ tracks the
source label of each vertex of $\graph$, according to the evaluation
of $t$ that builds $\graph$. We formalize this notion in general, for
graphs of arbitrary tree-width:

\begin{definition}\label{def:history}
  Given an optimal tree decomposition $(\tree,\bag)$ of a graph
  $\graph$ of tree-width $k$, a $k$-\emph{history of $(\tree,\bag)$}
  is a mapping $h : \vertof{\graph} \rightarrow \set{1,\ldots,k+1}$
  such that, for each node $n \in \vertof{T}$ and any two vertices
  $v\neq v' \in \bag(n)$, we have $h(v) \neq n(v')$.
\end{definition}

The following lemma is fairly standard and can be proved for arbitrary
tree decompositions, not just for the alternating ones. For reasons of
self-completeness, we prove a version of it below using our
terminology:

\begin{lemmaE}\label{lemma:history}
  Each optimal tree decomposition $(\tree,\bag)$ of a graph $\graph$
  of tree-width $k$ admits a $k$-history.
\end{lemmaE}
\begin{proofE}
  We build a $k$-history $h : \vertof{\graph} \rightarrow
  \set{1,\ldots,k+1}$ of $(\tree,\bag)$ by building pairwise disjoint
  sets $H_1, \ldots, H_{k+1}$ of vertices from $\graph$ and a set $M$
  of nodes from $\tree$, i.e., we set $h(v) = i \iff v \in H_i$, for
  all $1 \leq i \leq k+1$.

  Let $n \in \vertof{\tree}$ be a node such that $\cardof{\bag(n)} =
  k+1$. Such a node exists because $\twof{\graph}=k$ and
  $(\tree,\bag)$ is an optimal tree decomposition of
  $\graph$. Initially, we place each vertex from $\bag(n)$ into
  exactly one of the sets $H_1, \ldots, H_{k+1}$ and set
  $M=\set{n}$. The invariants of the construction
  are: \begin{compactenum}
  \item\label{it1:lemma:history} $\bigcup_{i=1}^{k+1} H_i = \bigcup_{m
    \in M} \bag(m)$, and
  \item\label{it2:lemma:history} $H_i \cap H_j = \emptyset$, for all
    $1 \leq i < j \leq k+1$.
  \end{compactenum}
  Note that these invariants are initially satisfied. We repeat the
  following step as long as $M \subsetneq \vertof{\tree}$. Let $m \in
  \vertof{\tree} \setminus M$ be a node having a neighbour $p \in M$
  such that $\set{v_1,\ldots,v_\ell} \isdef \adh{\tree}{m}{p}$. Let
  $H_{i_1}, \ldots, H_{i_\ell}$ be sets such that $v_j \in H_{i_j}$,
  for all $1 \leq j \leq \ell$. These sets are pairwise disjoint, by
  the fact that (\ref{it2:lemma:history}) holds before each
  application of the step. Each of the vertices from $\bag(m)
  \setminus \set{v_1,\ldots,v_\ell}$ is then placed into exactly one
  of the remaining sets $H_j$, for $j \in \set{1,\ldots,k+1} \setminus
  \set{i_1,\ldots,i_\ell}$ and $m$ is added to $M$.

  We prove that the invariants (\ref{it1:lemma:history}) and
  (\ref{it2:lemma:history}) hold after the application of the step.
  (\ref{it1:lemma:history}) holds because all vertices from $\bag(m)
  \setminus \set{v_1,\ldots,v_\ell}$ are placed into the sets $H_1,
  \ldots, H_{k+1}$ and these sets do not contain any other vertex of
  $\graph$ that is not in the bag of a node from $M \cup \set{m}$. To
  prove that (\ref{it2:lemma:history}) holds suppose, for a
  contradiction that the step has introduced a vertex $v \in H_i \cap
  H_j$, for some $1 \leq i < j \leq k+1$. Since $v$ is added to
  exactly one set, say $H_i$, it must have been present in $H_j$
  before the step was taken. Let $r \in M$ be the first encountered
  node such that $v \in \bag(r)$. Then $v$ is present in all bags
  along the path from $r$ to $m$ in $\tree$, hence also in $\bag(p)$,
  because there is a unique path between two nodes in a tree. But then
  $v \in \adh{\tree}{m}{p}$ already before the step and cannot be
  added to $H_i$ by the step, because only the vertices from $\bag(m)
  \setminus \adh{\tree}{m}{p}$ are added to $H_{i_1}, \ldots,
  H_{i_\ell}$ by the step, contradiction.   
\end{proofE}

\ifLongVersion
\paragraph{Proof of Theorem \ref{thm:triangle-graphs}}
``$\subseteq$'' Let $\graph \in \trianglegraphs$ be a triangle
graph. Since $\graph$ is a $3$-graph, we have
$\cardof{\vertof{\graph}} \geq 3$. If $\cardof{\vertof{\graph}} = 3$,
the only case is $\graph = \hval(\basictriangle)$, by the triangle
property. Otherwise, we have $\cardof{\graph} \geq 4$, hence $\graph$
has an inner vertex, call it $v$.

Since $\graph$ is a triangle graph, we have $2 \leq \twof{\graph} \leq
3$. Let us consider first the case $\twof{\graph} = 2$ and let
$(\tree,\bag)$ be an optimal tree decomposition of $\graph$. By adding
$v$ (i.e., the inner vertex of $\graph$ which exists due to
$\cardof{\graph} \geq 4$) to all bags, we transform $(\tree,\bag)$
into a tree decomposition of width $3$, to which one can apply the
construction from the proof of Lemma \ref{lemma:full-td} to obtain an
alternating tree decomposition having the $3$-step
property. Otherwise, if $\twof{\graph} = 3$, we apply directly Lemma
\ref{lemma:full-td} to obtain an optimal alternating tree
decomposition $(\tree,\bag)$ of $\graph$ having the $3$-step property.

In both cases, we can assume that $(\tree,\bag)$ is an alternating
tree decomposition of $\graph$ having the $3$-step property, with no
loss of generality. By Lemma \ref{lemma:history}, $(\tree,\bag)$
admits a $3$-history $h : \vertof{\graph} \rightarrow
\set{1,2,3,4}$. We shall build a ground triangle term $t$ such that
$\hval(t) = \graph$, by induction on the structure of $\tree$. By the
triangle property of $\graph$, each edge $xy\in\edgeof{\graph}$
belongs to a subgraph $xyz$ that is contained in some bag of
$(\tree,\bag)$\footnote{In general, each clique of a graph belongs to
some bag, for each tree decomposition of that graph.}. The invariant
of the construction is the following: for each node $n \in
\vertof{\tree}$ we build a ground triangle term $t_n$ such that
$\hval(t_n) = \graph[V_n]^\basictriangle$, where $V_n \isdef \bag(n)
\cup \bigcup_{p \text{ descendant of } n} \bag(p)$.

\begin{figure}[t!]
  \centerline{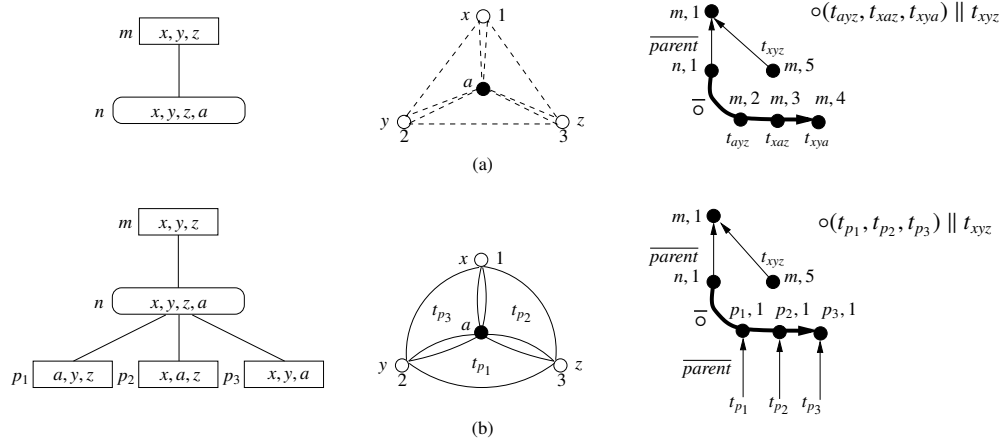}
  \caption{Translating alternating tree decompositions into triangle
    terms, in the base (a) and inductive (b) cases. The leftmost
    structures are the encodings of the corresponding triangle terms (\autoref{subsec:parsable}).}
  \label{fig:sop}
\end{figure}

For the base case, let $n$ be a leaf of $\tree$ and $m$ be its
adhesion parent. We recall that the leaves of an alternating tree
decomposition are bag nodes, hence their parents are adhesion nodes.
By the $3$-step property, $\cardof{\bag(n)} = 4$ and $\cardof{\bag(m)}
= 3$. Let $\bag(m) = \set{x,y,z}$ and $\bag(n) =
\set{x,y,z,a}$. Assume that $h(x) = 1$, $h(y) = 2$, $h(z) = 3$ and
$h(a)=4$ (the other cases are similar). We let $t_n \isdef
\sop(t_{ayz},t_{xaz},t_{xya})$, where each $t_{ijk}$ is either
$\basictriangle$ or $\emptygraph$, depending whether $ijk \in
\set{ayz,xaz,xya}$ is a subgraph of $\graph$. The order of the terms
in the serial composition depends on the histories of $x$, $y$ and
$z$, see Figure \ref{fig:sop} (a) for a depiction of the above case.
Since $h(x),h(y),h(z),h(a)$ is a permutation of $1,2,3,4$, there are
$24$ cases to consider overall.

For the inductive step, we distinguish two cases. If $n$ is an
adhesion node of $\tree$ and $p_1, \ldots, p_k$ are its bag children,
let $t_n \isdef t_{p_1} \pop \ldots \pop t_{p_k} \pop t_{xyz}$, where
$t_{p_1}, \ldots, t_{p_k}$ are already built triangle terms, by the
inductive hypothesis and $t_{xyz}$ is either $\basictriangle$ or
$\emptygraph$ depending on whether the vertices from $\bag(n) =
\set{x,y,z}$ induce a $K_3$ subgraph in $\graph$ or not.

Else, $n$ is a bag node of $\tree$ and let $m$ be its adhesion parent,
assumed w.l.o.g. to have the bags $\bag(m) = \set{x,y,z}$ and $\bag(n)
= \set{x,y,z,a}$. Assume furthermore that $h(x) = 1$, $h(y) = 2$, 
$h(z) = 3$ and $h(a)=4$ (the other cases are similar). Let $p_1, \ldots,
p_\ell$ be the adhesion children of $n$, where $1 \leq \ell \leq 3$.
We define $t_n = \sop(t_{ayz},t_{xaz},t_{xya})$, where each $t_{ijk}$
is defined below: \begin{enumerate}
\item $t_{ijk} \isdef t_{p_h}$ if $\bag(p_h) = \set{i,j,k}$, for some $1 \leq h \leq \ell$,
\item $t_{ijk} \isdef \basictriangle$ if the triangle
  $ijk$ is a subgraph of $\graph$ and $\bag(p_h) \neq \set{i,j,k}$, for all $1 \leq h \leq \ell$,
\item $t_{ijk} \isdef \emptygraph$ if none of the above applies.
\end{enumerate}
The order of terms in the serial composition depends on the histories
of $x$, $y$, $z$ and $a$, see Figure \ref{fig:sop} (b) for a
depiction of the construction in this case. Note that $t_{ijk}$ is
well-defined, because $\bag(t_{p_h}) \neq \bag(t_{p_g})$, for all $1
\leq h < g \leq \ell$, hence exactly one of the above cases matches,
for each $ijk \in \set{ayz,xaz,xya}$. It is easy to check that
$\hval(t_n) = \graph[V_n]^\basictriangle$ is an inductive invariant of
the construction. We obtain $\hval(t_r) = \graph$, where $r$ is the
root of $\tree$, therefore $\graph \in \trianglegraphs$ follows.

\vspace*{.5\baselineskip}
\noindent``$\supseteq$'' This follows from Lemma
\ref{lemma:witness-triangle} and $\twof{\hval(t)} \leq 3$, for each
ground triangle term $t$, by Lemma \ref{lemma:hr-tw}. \qed
\fi

\subsection{Algebraic Characterization of Fan Graphs}

We prove a similar characterization theorem for fan graphs, namely
that the set of evaluations of ground fan terms coincides with the set
of fan graphs (Definition \ref{def:fan-graph}): 

\begin{theorem}\label{thm:fan-graphs}
  $\trianglefangraphs = \set{\hval(t) \mid t \text{ is a ground
      fan term}}$.
\end{theorem}

Before giving the proof of the above theorem, we restate a well-known
result~\cite[Lemma 11.3]{DBLP:series/txtcs/FlumG06} about tree
decompositions in our terminology (Lemma
\ref{lemma:adhesion-separation}). A \emph{separation} of a graph
$\graph$ is a pair $(A,B) \in \pow{\vertof{\graph}} \times
\pow{\vertof{\graph}}$ such that $A \cup B = \vertof{\graph}$, $A
\setminus B \neq \emptyset$ and $B \setminus A \neq \emptyset$. When
$(A,B)$ is a separation of $\graph$, the set of vertices $\sep{A}{B}
\isdef A \cap B$ is called the \emph{separator} of $(A,B)$.

\begin{lemmaE}\label{lemma:adhesion-separation}
  Let $\graph$ be a graph, $(\tree,\bag)$ be a tree-decomposition of
  $\graph$, $nm \in \edgeof{\tree}$ be a tree edge and $N, M \subseteq
  \vertof{\tree}$ be the sets of nodes reachable from $n$ and $m$
  without crossing $nm$, respectively. Then $(\bigcup_{x\in N}
  \bag(x), \bigcup_{y\in M} \bag(y))$ is a separation of $\graph$ with
  separator $\adh{}{n}{m}$.
\end{lemmaE}
\begin{proofE}
  We denote $X\isdef\bigcup_{x\in N} \bag(x)$ and
  $Y\isdef\bigcup_{y\in M} \bag(y)$.  Since $N \uplus M =
  \vertof{\tree}$, we obtain $X \cup Y = \vertof{\graph}$. By
  Assumption \ref{ass:td}, we have $X \setminus Y \neq \emptyset$ and
  $Y \setminus X \neq \emptyset$, hence $(X,Y)$ is a separation of
  $\graph$. Since $X \cap Y$ is the separator of $(X,Y)$, it remains
  to prove that $X \cap Y = \adh{}{n}{m}$. ``$\supseteq$'' This
  direction is trivial, since $\bag(n) \subseteq X$ and $\bag(m)
  \subseteq Y$. ``$\subseteq$'' Suppose, for a contradiction, that
  there exists a vertex $x \in (X \cap Y) \setminus
  \adh{}{n}{m}$. Assume, moreover, that $x \in (X \cap Y) \setminus
  \bag(n)$ (the case $x \in (X \cap Y) \setminus \bag(m)$ is
  symmetric). Since $x \in X \setminus \bag(n)$, there exists a node
  $p \in N \setminus \set{n}$ such that $x \in \bag(p)$. Since $x \in
  Y$, there exists a node $r \in M$ such that $x \in \bag(r)$. Then,
  $x$ is in each bag on the path from $p$ to $r$, hence also in
  $\bag(n)$, contradiction.
\end{proofE}

The proof of Theorem \ref{thm:fan-graphs} uses the following version
of Lemma \ref{lemma:full-td}:

\begin{lemmaE}\label{lemma:connectivity-tw}
  Each optimal alternating tree decomposition of a $k$-connected
  $k$-graph of tree-width at most $k$ has the $k$-step property.
\end{lemmaE}
\begin{proofE}
  Let $\graph$ be a $k$-connected $k$-graph of tree-width $k$ at most
  and $n \in \vertof{\tree}$ be a node. Since $\twof{\graph}\leq k$,
  we have $\cardof{\bag(n)}\leq k+1$. If $n$ is an adhesion node, we
  distinguish two cases: \begin{compactitem}[-]
  \item $n$ is the root of $\tree$: in this case $\cardof{\bag(n)} =
    k$, by Definition \ref{def:alternating-td}
    (\ref{it4:def:alternating-td}).
  \item $n$ is not the root of $\tree$: in this case $\cardof{\bag(n)}
    \geq k$, because $\graph$ is $k$-connected and $\bag(n) =
    \adh{\tree}{n}{m}$ is a separator of $\graph$, by Lemma
    \ref{lemma:adhesion-separation}, where $m$ is the parent of $n$ in
    $\tree$.
  \end{compactitem}
  If $n$ is a bag node, then it is connected to some adhesion node,
  hence $\cardof{\bag(n)} = k+1$, by Assumption \ref{ass:td}. Also,
  $\cardof{\bag(n)}=k$ if $n$ is an adhesion node different from the
  root, because it is connected to a bag node, also by Assumption
  \ref{ass:td}.
\end{proofE}

The right to left direction of Theorem~\ref{thm:fan-graphs} uses
Lemma~\ref{lemma:fan-connectivity}, which is based on the following
lemma:

\begin{lemmaE}[Connectivity Lemma]\label{lemma:connectivity-lemma}
  Let $\graph$ be a connected graph and let $\graph_1,\ldots,\graph_m$
  be 3-graphs whose vertices and edges are contained in $\graph$.
  Assume that (1) every edge of $\graph$ is also an edge of some
  $\graph_i$, and (2) for every $i$, the thickening
  $\overline{\graph_i} = \graph_i \pop \basictriangle$ is
  3-connected. Moreover, let $B \subseteq \vertof{\graph}$ be a set
  such that (3) every source vertex of every $\graph_i$ belongs to
  $B$, and (4) for every $i \neq j$, $\vertof{\graph_i} \cap
  \vertof{\graph_j} \subseteq B$. Let $X \subseteq \vertof{\graph}$
  with $\cardof{X} \le 2$.  If all vertices of $B \setminus X$ lie in
  one connected component of $\graph - X$, then $\graph - X$ is
  connected.
\end{lemmaE}
\begin{proofE}
  Suppose for contradiction that $\graph - X$ is disconnected. Since
  all vertices of $B \setminus X$ lie in one connected component of
  $\graph - X$, there must be a connected component $C$ of $\graph -
  X$ disjoint from $B \setminus X$. Since every edge belongs to some
  $\graph_i$ and the different subgraphs intersect only in $B$, the
  component $C$ is contained entirely in some $\graph_j - X$.
  Moreover, $C$ avoids all sources of $\graph_j$, because all sources
  of $\graph_j$ are contained in $B$. Therefore, adding edges between
  the sources of $\graph_j$ cannot connect $C$ to the rest of
  $\graph_j - X$. Hence, $C$ is also a connected component of
  $\overline{\graph_j}$, contradicting the 3-connectedness of
  $\overline{\graph_j}$.
\end{proofE}

The following lemma characterizes the connectivity of graphs that are
values of fan terms: 

\begin{lemmaE}\label{lemma:fan-connectivity}
  Let $t$ be a ground fan term and $\graph \isdef \hval(t)$ be its
  value. Then, $\graph$ is $2$-connected, in general, and
  $3$-connected if only if $t = \sop(t_1,t_2,t_3)$ or $t = t_1 \pop
  t_2$, for some ground fan terms $t_1$, $t_2$ and $t_3$. Moreover,
  the thickening $\overline\graph \isdef \graph \pop
  \hval(\basictriangle)$ of $\graph$ is $3$-connected if $\graph$ is
  not $K_3$. 
\end{lemmaE}
\begin{proofE}
We use the following fact, whose easy proof is left to the reader:

\begin{fact}\label{fact:glue}
  Let $\graph_1$ and $\graph_2$ be two $k$-connected graphs and
  $\graph$ be the graph obtained by glueing at least $k$ vertices from
  both $\graph_1$ and $\graph_2$. Then $\graph$ is $k$-connected.
\end{fact}

We prove that $\graph$ is $2$-connected by induction on the structure
of $t$, considering the following cases: \begin{compactitem}[-]
\item $t = \basictriangle$: in this case $\graph$ is a $K_3$ graph,
  hence it is $2$-connected.
\item $t = t_1 \fan{i} t_2$, for some $1 \leq i \leq 3$: we denote
  $\graph_j \isdef \hval(t_j)$, for both $j=1,2$, hence $\graph =
  \graph_1 \fan{i} \graph_2$. The two graphs are glued along two
  vertices $a$ and $s$, where $s$ remains a source of $\graph$, while
  $a$ becomes an inner vertex. Let $p$ and $q$ be the other sources of
  $\graph_1$ and $\graph_2$, respectively. The sources of $\graph$ are
  $s$, $p$ and $q$. Since $\graph_1$ and $\graph_2$ are $2$-connected,
  by the induction hypothesis, and are glued along $a$ and $s$, the
  graph is $2$-connected by Fact \ref{fact:glue}.
\item $t = \sop(t_1,t_2,t_3)$: we denote $\graph_j \isdef \hval(t_j)$,
  for $j=1,2,3$, hence $\graph =
  \sop(\graph_1,\graph_2,\graph_3)$. Since $\graph_1$, $\graph_2$ and
  $\graph_3$ are 2-connected, by the induction hypothesis, and each
  pair $(\graph_j,\graph_k)$, for $1 \leq j < k \leq 3$ is glued along
  two vertices, $\graph$ is 2-connected by Fact \ref{fact:glue}. This
  can be seen by first gluing $\graph_1$ with $\graph_2$, and then
  gluing the resulting graph with $\graph_3$.
\item $t = t_1 \pop t_2$: we denote $\graph_j \isdef \hval(t_j)$, for
  both $j=1,2$, hence $\graph = \graph_1 \pop \graph_2$. Since
  $\graph_1$ and $\graph_2$ are $2$-connected, by the induction
  hypothesis, and are glued along three vertices, the graph is
  $2$-connected by Fact \ref{fact:glue}.
\end{compactitem}

We prove that $\graph$ is $3$-connected if $t = \sop(t_1,t_2,t_3)$ or
$t = t_1 \pop t_2$, for some ground fan terms $t_1$, $t_2$ and $t_3$,
by considering the two cases below: \begin{compactitem}[-]
\item $t = \sop(t_1,t_2,t_3)$: let $\graph_j \isdef \hval(t_j)$, for
  each $j = 1,2,3$, $B_1 \isdef \set{a,y,z}$, $B_2 \isdef
  \set{a,x,z}$, $B_3 \isdef \set{a,x,y}$ be the sources of the graphs
  $\graph_1$, $\graph_2$ and $\graph_3$, respectively, and $B \isdef
  \set{a,x,y,z}$. Fix $X \subseteq \vertof{\graph}$ with $\cardof{X}
  \le 2$.  By Lemma \ref{lemma:connectivity-lemma}, it suffices to
  show that $B \setminus X$ is connected in $G - X$. Choose $p,q \in
  \set{1,2,3}$ such that $\cardof{\vertof{\graph_p}} -
  \cardof{\vertof{\graph_p - X}} \leq 1$, $\cardof{\vertof{\graph_q}}
  - \cardof{\vertof{\graph_q - X}} \leq 1$ and $(B_p \cap B_q)
  \setminus X \neq \emptyset$. Note that such a choice exists because
  $\cardof{X} \le 2$. Then $\graph_p - X$ and $\graph_q - X$ are
  connected, intersect in a surviving source, and together contain all
  of $B \setminus X$.  Hence, $\graph - X$ is connected and
  $\overline\graph$ is $3$-connected, by Lemma
  \ref{lemma:connectivity-lemma}.
\item $t = t_1 \pop t_2$, where $\graph_j = \hval(t_j)$, for
  $j=1,2$. Let $S = \{s_1,s_2,s_3\}$ be the common source set. Fix $X
  \subseteq \vertof{\graph}$ with $\cardof{X} \le 2$. For $j=1,2$,
  consider the connected components of $\graph_j - X$. We claim that
  each such component contains a vertex from $S \setminus X$.
  Otherwise, such a component would avoid all surviving sources, and
  adding edges between sources would not connect it to the rest.
  Hence it would also be a component of $\overline{\graph_j} - X$,
  contradicting the 3-connectedness of $\overline{\graph_j}$. If $S
  \setminus X$ has cardinality one, then all components of $\graph -
  X$ meet in that source, so $\graph - X$ is connected.  Otherwise,
  there are at least two sources in $S \setminus X$.  Since
  $\cardof{X} \le 2$, at least one of the sides, say $\graph_1$, 
  contains at most one vertex of $X$. By 2-connectedness of
  $\graph_1$, the graph $\graph_1 - X$ is connected. In particular, it
  connects all sources of $S \setminus X$. Since all components of
  $\graph - X$ are connected to some source, $\graph - X$ is
  connected.
\end{compactitem}
We prove that $\graph$ is $3$-connected only if $t =
\sop(t_1,t_2,t_3)$ or $t = t_1 \pop t_2$, for some ground fan terms
$t_1$, $t_2$ and $t_3$. Since $\graph$ is $3$-connected, $t \neq
\basictriangle$, because $K_3$ is not $3$-connected. Suppose, for a
contradiction, that $t = t_1 \fan{i} t_2$, for some $1 \leq i \leq 3$
and some ground fan terms $t_1$ and $t_2$. Let $\graph_j \isdef
\hval(t_j)$, for $j = 1,2$ and $p$, $q$ be the vertices of $\graph$
obtained by joining the corresponding sources $p_1$, $q_1$ of
$\graph_1$ and $p_2$, $q_2$ of $\graph_2$, respectively. Because
$\vertof{\graph_1}$ and $\vertof{\graph_2}$ are disjoint,
$((\vertof{\graph_1} \setminus \set{p_1,q_1}) \cup \set{p,q},
(\vertof{\graph_2} \setminus \set{p_2,q_2}) \cup \set{p,q})$ is a
$2$-separation of $\graph$. This is because each $\graph_i$ has a
third source different from $p_i$, $q_i$. However, the existence of a
$2$-separation contradicts the fact that $\graph$ is $3$-connected.

Finally, we prove that $\overline{\graph}$ is $3$-connected if $t \neq
\basictriangle$. If $t = \sop(t_1,t_2,t_3)$ or $t = t_1 \pop t_2$, for
some ground fan terms $t_1$, $t_2$ and $t_3$, then $\graph$ is already
$3$-connected and there is nothing to prove. Otherwise, $t = t_1
\fan{i} t_2$, for some $1 \leq i \leq 3$. Let $B = \{a,s,p,q\}$, where
$a$, $s$, $p$ and $q$ are the vertices considered in the $t = t_1
\fan{i} t_2$ case above. Fix $X \subseteq \vertof{\graph}$ with
$\cardof{X} \le 2$. By Lemma \ref{lemma:connectivity-lemma}, it
suffices to show that $B \setminus X$ is connected in $\overline\graph
- X$. The vertices $spq$ form a triangle in $\overline\graph$, so the
vertices from $\set{s,p,q} \setminus X$ are connected in
$\overline\graph - X$. If $a \not\in X$ then one of the graphs
$\graph_i - X$ contains $a$ and another of its sources not in
$X$. Because $\graph_i$ is $2$-connected, $a$ is connected to this
other source. Hence, $B \setminus X$ is connected and
$\overline\graph$ is $3$-connected, by Lemma
\ref{lemma:connectivity-lemma}.
\end{proofE}
\ifLongVersion
\paragraph{Proof of Theorem \ref{thm:fan-graphs}}
``$\subseteq$'' Let $\graph$ be a fan graph. If $\graph$ is $K_3$ then
$\graph=\hval(\basictriangle)$ and there is nothing left to
prove. Hence, we assume in the following that $\graph \neq K_3$. Let
$(\tree,\bag)$ be an optimal alternating tree decomposition of
$\graph$. By Lemma \ref{lemma:alternating-td}, such a decomposition
exists. Let $\overline{\graph} \isdef \graph \pop
\hval(\basictriangle)$ be the thickening of $\graph$. Then,
$(\tree,\bag)$ is a tree decomposition of $\overline{\graph}$ as
well. Because $\overline{\graph}$ is $3$-connected, $(\tree,\bag)$ has
the $3$-step property, by Lemma \ref{lemma:connectivity-tw}, and
admits a $3$-history $h : \vertof{\graph} \rightarrow \set{1,2,3,4}$,
by Lemma \ref{lemma:history}.

\begin{figure}[t!]
  \centerline{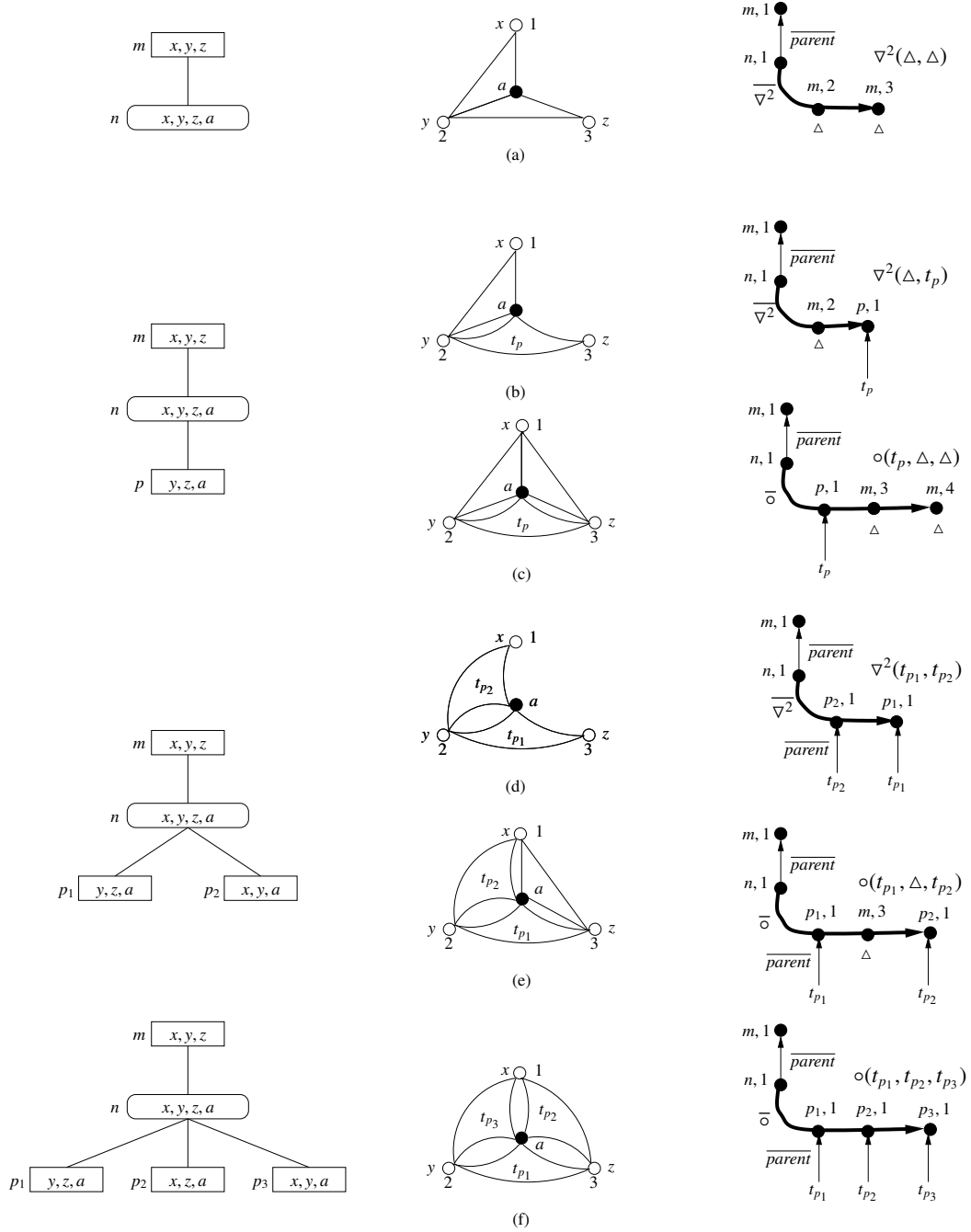}
  \caption{Translating alternating tree decompositions into fan
    terms, in the base (a) and inductive (b,c,d,e,f) cases. The leftmost
    structures are the encodings of the corresponding triangle terms (\autoref{subsec:parsable}).}
  \label{fig:fan}
\end{figure}

We build a ground triangle term $t$ such that $\hval(t) =
\overline{\graph}$, by induction on the structure of $\tree$. For each
node $n \in \vertof{\tree}$, we build a ground triangle term $t_n$,
such that: \begin{compactenum}[(1)]
\item the sources of $\hval(t_n)$ belong to $\bag(n)$,
\item if $n$ is an adhesion node then $\hval(t_n)$ is the triangle
  subgraph of $\overline{\graph}$ induced by the bags of the
  descendants of $n$ in $\tree$,
\item if $n$ is a bag node then the triangle subgraph of
  $\overline{\graph}$ induced by the bags of the descendants of $n$ is
  either $\hval(t_n \pop \basictriangle)$ or $\hval(t_n)$, the choice
  between $\hval(t_n \pop \basictriangle)$ and $\hval(t_n)$ being
  determined by whether $\overline{\graph}$ contains a triangle within the
  vertices $\bag(m)$ for the adhesion parent $m$ of $n$.
\end{compactenum}

For the base case, assume that $n$ is a (bag) leaf of $\tree$ and let
$m$ be its adhesion parent. Let $a$ be the unique vertex from $\bag(n)
\setminus \bag(m)$. Since $\overline{\graph}$ is $3$-connected, $a$
must be connected to the other $3$ vertices from $\bag(n)$, call them
$x$, $y$ and $z$.  Then, there must be edges $xa, ya, za \in
\edgeof{\graph}$.  Because $\overline{\graph}$ has the triangle
property, each of these edges must have a witness. Since $a$ does not
belong to any bag other than $\bag(n)$, each witness must be from
$\bag(n)$. By the pigeonhole principle, two of the edges $xa$, $ya$
and $za$ must have the same witness in $\bag(n)$. Assume that $y$
witnesses both $xa$ and $za$, the other cases being similar. We define
$t_n \isdef \basictriangle \fan{i} \basictriangle$, where the choice
of $i$ depends on the values of $h(x)$, $h(y)$, $h(z)$ and $h(a)$. See
Figure \ref{fig:fan} (a) for a depiction of this case, where $h(x)=1$,
$h(y)=2$, $h(z)=3$ and $h(a)=4$.  Let $V_n \isdef \bag(n) \cup
\bigcup_{p \text{ descendant of } n} \bag(p)$. It is easy to verify
that either $\hval(t_n \pop
\basictriangle)=\overline{\graph}[V_n]^\basictriangle$ or
$\hval(t_n)=\overline{\graph}[V_n]^\basictriangle$, depending on
whether $xyz$ is a subgraph of $\overline{\graph}$ or not,
respectively.

For the inductive case, let $n\in\vertof{\tree}$ be a node. We
distinguish two cases, depending on whether $n$ is an adhesion or a
bag node. If $n$ is an adhesion node, a term $t_p$ has been already
built for each bag child $p$ of $n$, by the inductive
hypothesis. Then, we define:
\begin{align*}
  t_n \isdef \begin{cases}
    \Pop_{p \text{ child of } n} t_p \pop \basictriangle &
    \text{if } \overline{\graph}[\bag(m)] = K_3 \\
    \Pop_{p \text{ child of } n} t_p & \text{otherwise}
  \end{cases}
\end{align*}
It is easy to verify that $\hval(t_n)= \overline{\graph}[\bag(n)]^\basictriangle$
in this case.

Otherwise, $n$ is a bag node and let $m$ be the adhesion parent of
$n$, such that $\bag(m) = \set{x,y,z}$ and $\bag(n) =
\set{x,y,z,a}$. We distinguish three cases: \begin{itemize}
\item $n$ has a single adhesion child $p$: in this case
  $\bag(m)\neq\bag(p)$, by Definition \ref{def:alternating-td}
  (\ref{it3:def:alternating-td}) and let $\bag(p)=\set{y,z,a}$ (the
  other cases are similar). Since $\overline{\graph}$ is
  $3$-connected, $x$ must be connected to $y$, $z$ and $a$. In
  particular, the path between $x$ and $a$ consists of a single edge
  $xa$, because $n$ is the only node in $\tree$ whose bag contains
  both $x$ and $a$. Because $\overline{\graph}$ has the triangle
  property, the edge $xa$ must have witness $y$ or $z$. We consider
  three cases (with the first point detailing two
  cases): \begin{itemize}
  \item Assume that $xa$ has witness $y$ but no witness $z$ (the case
    of witness $z$ but no witness $y$ is symmetric).  Hence there are
    edges $xy$ and $ya$ in $\overline{\graph}$.  We define $t_n$ as a $\fan{i}$
    operation between $\basictriangle$ (for the $xya$ subgraph) and
    $t_p$, where $i$ and the order of the arguments are uniquely
    determined by the values of $h(x)$, $h(y)$, $h(z)$ and $h(a)$. For
    instance, if $h(x)=1$, $h(y)=2$, $h(z)=3$ and $h(a)=4$ then $t_n
    \isdef \basictriangle \fan{2} t_p$. See Figure \ref{fig:fan} (b)
    for a depiction of this case.
  \item Assume $y$ and $z$ both witness the edge $xa$. Hence there are
    edges $xy,xz,ya$ and $za$ in $\overline{\graph}$. We define $t_n$ as a serial
    composition between two $\basictriangle$ (one for $xya$ and the
    other for $xza$) and $t_p$, the order of the arguments being
    determined by $h(x)$, $h(y)$, $h(z)$ and $h(a)$. For instance, if
    $h(x)=1$, $h(y)=2$, $h(z)=3$ and $h(a)=4$ then $t_n \isdef
    \sop(t_p,\basictriangle,\basictriangle)$. See Figure \ref{fig:fan}
    (c) for a depiction of this case.
  \end{itemize}
\item $n$ has two adhesion children $p_1$ and $p_2$: in this case
  $\bag(m)\neq\bag(p_i)$, for $i=1,2$, by Definition
  \ref{def:alternating-td} (\ref{it3:def:alternating-td}) and let
  $\bag(p_1)=\set{y,z,a}$ and $\bag(p_2)=\set{x,y,a}$ (the other cases
  are symmetric).
  We consider two cases: \begin{itemize}
  \item $G$ does not contain the triangle $xza$: in this case, we
    define $t_n$ as a $\fan{h(y,m)}$ operation between $t_{p_1}$ and
    $t_{p_2}$, the order of the arguments being determined by $h(x)$,
    $h(y)$, $h(z)$ and $h(a)$. See Figure \ref{fig:fan} (d) for a
    depiction of this case, where $h(x)=1$, $h(y)=2$, $h(z)=3$ and
    $h(a)=4$.
  \item $G$ contains the triangle $xza$: in this case, we define $t_n$
    as a serial composition between $t_{p_1}$, $t_{p_2}$ and
    $\basictriangle$ (the $xza$ subgraph), the order of the arguments
    being determined by $h(x)$, $h(y)$, $h(z)$ and $h(a)$. See Figure
    \ref{fig:fan} (e) for a depiction of this case, where $h(x)=1$,
    $h(y)=2$, $h(z)=3$ and $h(a)=4$.
  \end{itemize}
\item $n$ has three adhesion children $p_1$, $p_2$ and $p_3$: In this
  case let $t_n$ be the serial composition of $t_{p_1}$, $t_{p_2}$ and
  $t_{p_3}$, the order of the arguments being determined by $h(x)$,
  $h(y)$, $h(z)$ and $h(a)$. See Figure \ref{fig:fan} (f) for a
  depiction of this case, where $h(x)=1$, $h(y)=2$, $h(z)=3$ and
  $h(a)=4$.
  \end{itemize}
Note that these are the only cases possible, because if $n$ had four
adhesion children, then necessarily one of them would have the same
bag as $m$, contradicting Definition \ref{def:alternating-td}
(\ref{it3:def:alternating-td}). It is easy to verify in all cases
that either $\overline{\graph}[\bag(n)]^\basictriangle = \hval(t_n \pop
\basictriangle)$ or $\overline{\graph}[\bag(n)]^\basictriangle = \hval(t_n)$,
depending on whether $xyz$ is a subgraph of $\overline{\graph}$, respectively.

Finally, since the sources of $\graph$ form a triangle subgraph of
$\overline{\graph}$, we obtain $\overline{\graph} = \hval(t \pop
\basictriangle)$, for a ground fan term $t$. Hence, we
obtain: \begin{align*}
  \graph = \begin{cases}
    \hval(t\pop \basictriangle) & \text{if } \graph[\set{\sourceof{\graph}(1), 
        \sourceof{\graph}(2),\sourceof{\graph}(3)}] = K_3 \\
    \hval(t) & \text{otherwise}
  \end{cases}
\end{align*}
In both cases $\graph$ is the value of a ground fan term. 

\vspace*{.5\baselineskip}\noindent''$\supseteq$'' This follows from
Lemma \ref{lemma:fan-connectivity} and $\twof{\hval(t)} \leq 3$, for
any ground triangle term $t$, by Lemma \ref{lemma:hr-tw}. \qed
\fi

\section{Sets of Triangle Graphs}
\label{sec:sets}

In this section, we use the algebras introduced in
\autoref{sec:algebras} to describe sets of triangle graphs. In doing
so, we consider \emph{recognizable languages}, i.e., sets defined by
homomorphisms into finite algebras having the same signature, and
\emph{context-free languages}, i.e., sets defined by grammars. We
establish the relations between contex-free and recognizable languages
in \hr, the triangle and fan algebras, respectively, see Figure
\ref{fig:sets} for a chart of the relations between the different
classes of graph languages. In particular, we show that
recognizability in the triangle and fan algebras coincides with
definability in Counting Monadic Second Order (\cmso) logic, for the
respective classes of graphs. As a byproduct, we obtain the
decidability of \cmso{} over context-free sets of triangle and fan
graphs, thus motivating the applications of triangle and fan algebras
in domains such as verification, synthesis and learning.

\subsection{Recognizable Sets}
\label{subsec:recognizable}

We recall the definition of recognizable sets, in general, for any
subset of the domain of an algebra.  Let $\signature = \set{f_1, f_2,
  \ldots}$ be a signature of function symbols of arities $\arityof{f}$
and $\algof{A} = (\universeOf{A}, \set{f^\algof{A}}_{f\in\signature})$
be an $\signature$-algebra, where the interpretation of a function
symbol $f \in \signature$ in the domain $\universeOf{A}$ is the
function $f^\algof{A} : \universeOf{A}^{\arityof{f}} \rightarrow
\universeOf{A}$ (here $\universeOf{A}^n$ denotes the $n$-times
Cartesian product of $\universeOf{A}$ with itself). An algebra
$\algof{A}$ is \emph{finite} if its domain $\universeOf{A}$ is finite.
The interpretation of a ground term $t$ in $\algof{A}$ is the unique
element of the algebra denoted $t^\algof{A} \in \universeOf{A}$ which
is the result of interpreting the function symbols in $t$ according to
$\algof{A}$ and computing the result.

A \emph{homomorphism} between algebras $\algof{A}$ and $\algof{B}$,
having the same signature $\signature$, is a function $h :
\universeOf{A} \rightarrow \universeOf{B}$ such that
$h(f^\algof{A}(a_1,\ldots,a_{\arityof{f}})) = f^\algof{B}(h(a_1),
\ldots, h(a_{\arityof{f}}))$, for all $f \in \signature$ and
$a_1,\ldots,a_{\arityof{f}} \in \universeOf{A}$. A \emph{recognizable
set} does not distinguish between elements having the same homomorphic
image in a finite algebra:

\begin{definition}\label{def:recognizability}
  A set $\langu \subseteq \universeOf{A}$ is \emph{recognizable} in an
  $\signature$-algebra $\algof{A} = (\universeOf{A},
  \set{f^\algof{A}}_{f\in\signature})$ if and only if there exists a
  finite algebra $\algof{B} = (\universeOf{B},
  \set{f^\algof{B}}_{f\in\signature})$ and a homomorphism $h$ between
  $\algof{A}$ and $\algof{B}$ such that $\langu=h^{-1}(F)$ for some
  subset $F \subseteq \universeOf{B}$.
\end{definition}




\subsection{Context-free Sets}
\label{subsec:context-free}

A context-free set is a component of the least solution of a finite
set of recursive definitions. Formally, given a signature
$\signature$, a \emph{grammar} is a tuple
$\grammar=(\nonterm,\rules,\axioms)$, where $\nonterm$ is a finite set
of \emph{nonterminals} ranged over by $X,Y,\ldots$, $\rules$ is a
finite set of \emph{rules} of the form $X \rightarrow
t[Y_1,\ldots,Y_n]$, where $t$ is a term over $\signature$, and
$\axioms \subseteq \nonterm$ is a set of \emph{axioms}.

A \emph{derivation step} is a pair of terms, denoted $\theta
\step{\grammar} \eta$, such that $\eta$ is obtained from $\theta$ by
replacing an occurrence of a nonterminal $X$ with the right-hand side
$t$ of a rule $X \rightarrow t \in \rules$. A \emph{derivation}
$\theta \step{\grammar}^* \eta$ is a finite sequence of steps. The
derivation is complete if $\eta$ is a ground term. Given an algebra
$\algof{A}$ over the signature $\signature$, the \emph{language} of
$\grammar$ in $\algof{A}$ is the set $\langof{A}{\grammar} \isdef
\set{\theta^\algof{A} \mid X \step{\grammar} \theta \text{ is a
    complete derivation, } X \in \axioms}$ of interpretations in
$\algof{A}$ of the ground terms produced by $\grammar$ starting with
an axiom. A set is \emph{context-free} if it is the language of some
grammar.

We say that a $\signature_\algof{B}$-algebra $\algof{B}$ is
\emph{derived} from the $\signature_\algof{A}$-algebra $\algof{A}$ if,
for every function symbol $f \in \signature_\algof{B}$, there exists
some $\signature_\algof{A}$-term $t_f$ such that $f^\algof{B} =
t_f^\algof{A}$. By the above definition, each context-free set in a
derived algebra $\algof{B}$ is context-free in the base algebra
$\algof{A}$. This is because one can replace, in the right-hand sides
of the rules of a grammar in the signature of $\algof{B}$, the derived
function symbols with their definitions, thus obtaining a grammar in
the signature of $\algof{A}$, that produces the same language. The two
inclusions on the second line in Figure \ref{fig:sets} are immediate
consequences of the fact that the fan algebra is derived from the
triangle algebra, which is, in turn, derived from the \hr{} algebra.

Unlike in the case of words, where each recognizable languages is
context-free, for the graphs defined in the \hr{} algebra, there is no
such inclusion. In fact, neither of classes of \hr{} context-free and
\hr{} recognizable sets is included in the other. The following
result, known as the \emph{Filtering Theorem}, generalizes a similar
result for words, thus making the relation between recognizable and
context-free sets formal, in general:

\begin{theorem}[Theorem 3.88 in~\cite{courcelle_engelfriet_2012}]\label{thm:ft}
  Let $\algof{A}$ be an algebra and $\langu,\klangu$ be subsets of the
  domain of $\algof{A}$. If $\langu$ is recognizable in $\algof{A}$
  and $\klangu$ is context-free in $\algof{A}$ then
  $\langu\cap\klangu$ is context-free in $\algof{A}$. Moreover, a
  grammar that produces $\langu\cap\klangu$ can be effectively built
  from the finite algebra that recognizes $\langu$ and the grammar
  that produces $\klangu$.
\end{theorem}

A consequence of this theorem is that, for the triangle and fan
algebras, each recognizable language is context-free, thus
establishing the two vertical inclusions from Figure \ref{fig:sets}:

\begin{corollaryE}\label{cor:rec-cf}
  Each recognizable set of triangle (resp. fan) graphs is context-free
  in the triangle (resp. fan) algebra.
\end{corollaryE}
\begin{proofE}
  By Theorem \ref{thm:triangle-graphs}, the grammar
  $\grammar_{\trianglealg}$, having following rules, produces the set
  of triangle graphs, i.e.,
  $\langof{\trianglealg}{\grammar_{\trianglealg}} =
  \trianglegraphs$. The only nonterminal of $\grammar_{\trianglealg}$
  is $X$, which is also its axiom.
  \begin{align*}
    X \rightarrow & ~X \pop X \\
    X \rightarrow & ~\sop(X,X,X) \\
    X \rightarrow & ~\basictriangle \\
    X \rightarrow & ~\emptygraph
  \end{align*}
  Let $\langu$ be a recognizable set in the triangle algebra. By
  Theorem \ref{thm:ft}, there exists a grammar $\grammar$ such that
  $\langof{\trianglealg}{\grammar} =
  \langof{\trianglealg}{\grammar_{\trianglealg}} \cap \langu =
  \trianglegraphs \cap \langu = \langu$. The proof is similar for fan
  graphs.
\end{proofE}
\noindent As a side remark, a similar result holds for context-free
sets of graphs recognizable in any fragment of \hr{} that uses
finitely many source labels.

\subsection{\cmso-definable Sets}
\label{subsec:cmso}

Another way of describing sets of graphs is using the Counting Monadic
Second Order (\cmso) logic. This method of graph specification is
fairly different from recognizability and context-freeness, because it
does not require an underlying algebra. Moreover, \cmso{} can express
many non-trivial properties of graphs, such as $k$-colorability,
$k$-connectivity, the existence of Eulerian and Hamiltonian cycles and
planarity. In particular, any set of graphs defined by a constant
bound on their tree-width can be defined in \cmso, due to the
celebrated Graph Minor Theorem of Seymour and
Robertson~\cite{SEYMOUR199322}. For instance, the set of graphs
$\set{\graph \mid \twof{\graph}\leq3}$ is defined by four known
excluded minors~\cite{ARNBORG19901}, meaning furthermore that the
\cmso{} formula that defines it is effectively computable\footnote{The
presence of a given minor in a graph is yet another \cmso-definable
property.}.

We recall below the definition of \cmso{} for two purposes. First, we
prove that the satisfiability problem for \cmso{} is decidable over
context-free sets of triangle and fan graphs, respectively. This is a
consequence of the fact that triangle and fan graphs have algebraic
characterizations (Theorems \ref{thm:triangle-graphs} and
\ref{thm:fan-graphs}) and these algebras are derived from \hr{},
following a well-known result of Courcelle, which states that \cmso{}
is decidable over \hr{} context-free sets~\cite[Theorem
  1.22]{courcelle_engelfriet_2012}. Second, we prove that the
\cmso-definable sets are the same as the sets recognizable in both the
triangle and fan algebras (\autoref{subsec:parsable}). This
specializes to the triangle algebras a known result on the equivalence
between recognizability in \hr{} and \cmso-definability of sets of
graphs of bounded tree-width~\cite{10.1145/2933575.2934508}.

Let $\relations$ be a finite signature of \emph{relation symbols}
$\arel \in \relations$, of arities $\arityof{\arel} \geq 1$. As usual,
relation symbols of arity $1$, $2$ and $3$ are called unary, binary
and ternary, respectively. The \cmso{} logic is the set of formul{\ae}
defined inductively by the following syntax:
\begin{align*}
  \psi := & \ x=y \mid \arel(x_1,\ldots,x_{\arityof{\arel}}) \mid X(x)
  \mid \cardconstr{X}{q}{p} \mid \psi \land \psi \mid \neg\psi \mid
  \exists x ~.~ \psi \mid \exists X ~.~ \psi
\end{align*}
where $x,y_1,\ldots,y_{\arityof{a}} \in \vars$ are \emph{individual
  variables}, $X \in \Vars$ are \emph{set variables} and $p,q \in
\nat$ are integers such that $p \in \interv{0}{q-1}$. A
\emph{sentence} is a formula in which each variable occurs in the
scope of a quantifier. We denote by \mso{} the fragment of \cmso{}
without the cardinality constraints of the form
$\cardconstr{X}{q}{p}$.

The semantics of \cmso{} is given in terms of structures. A
$\relations$-\emph{structure} (or simply a structure, when
$\relations$ is understood) is a tuple
$\astruc=(\univ,\set{\arel^\astruc}_{\arel\in\relations})$, where
$\univ$ is a set called \emph{universe} and $\arel^\astruc \subseteq
\univ^{\arityof{\arel}}$ is the \emph{interpretation} of the relation
symbol $\arel$ in $\astruc$, i.e., a relation over $\univ$ of the same
arity as $\arel$. Note the similarity between structures and algebras,
the first being interpretations of relation symbols, whereas the
latter interpreting function symbols. We denote by
$\strucof{\relations}$ the set of $\relations$-structures.

The \emph{satisfaction} relation $\astruc \models^\store \psi$ between
structures and \cmso{} formul{\ae} is parameterized by a mapping
$\store$ of individual (resp. set) variables to (resp. sets of)
vertices. The relation is defined inductively on the structure of
formul{\ae}, as follows:
\[\begin{array}{lclclcl}
\astruc \models^\store x = y & \iff & \store(x)=\store(y) &&
\astruc \models^\store \cardconstr{X}{q}{p} & \iff & \cardof{\store(X)} = p \mod q \\
\astruc \models^\store \arel(x_1,\ldots,x_{\arityof{\arel}}) & \iff & (\store(x_1), \ldots, \store(x_{\arityof{\arel}})) \in \arel^\astruc &&
\astruc \models^\store \exists x ~.~ \psi & \iff & \graph \models^{\store[x \leftarrow u]} \psi \text{, for some } u \in \univ \\
\graph \models^\store X(x) & \iff & \store(x) \in \store(X) &&
\astruc \models^\store \exists X ~.~ \psi & \iff & \graph \models^{\store[X \leftarrow U]} \psi \text{, for some } U \subseteq \univ
\end{array}\]
Here $\store[x\leftarrow u]$ (resp. $\store[X\leftarrow U]$) denotes
the store that agrees with $\store$ everywhere except in $x$
(resp. $X$) where it equals $u$ (resp. $U$). The semantics of boolean
conjunction and negation are standard, omitted for brevity. If $\phi$
is a sentence, we write $\astruc \models \phi$ instead of $\astruc
\models^\store \phi$.

There are several ways in which graphs can be represented by
structures. For technical reasons that will be made clear below, we
adopt here the \emph{incidence model}, in which the universe of the
structure $\strucof{\graph}$ that encodes the graph $\graph$ consists
of the vertices and edges of $\graph$ and the incidence relation
$\edgrel(x,y,z)$ means that $x$ is an edge attached to the vertices
$y$ and $z$. In addition, we consider unary relations $\srcrel{i}(x)$
meaning that $x$ is the $i$-th source of $\graph$. Formally, we
define:
\begin{align*}
  \strucof{\graph} \isdef \Big(\vertof{\graph}\uplus\edgeof{\graph},
  \underbrace{\set{(xy,x,y) \mid xy \in \edgeof{\graph}}}_{\edgrel^{\strucof{\graph}}},
  \Big\{\underbrace{\set{\sourceof{\graph}(i)}}_{\srcrel{i}^{\strucof{\graph}}}\Big\}_{i\in\dom{\sourceof{\graph}}}\Big)
\end{align*}
Note that the distinction between vertices and edges is achieved via
the \mso{} formul{\ae} $\mathit{isedge}(x) \isdef \exists y \exists z
~.~ \edgrel(x,y,z)$ and $\mathit{isvert}(x) \isdef
\neg\mathit{isedge}(x)$, hence one can unambiguously quantify
separately over vertices and edges, respectively. Moreover, we assume
that the following sentence holds for all structures that encode graphs:
\begin{align*}
  \forall e \forall f \forall x \forall y \forall z \forall u ~.~ \edgrel(e,x,y) \wedge \edgrel(f,z,u) \wedge e=f \rightarrow (x=z \wedge y=u \vee x=u \wedge y=z)
\end{align*}

A set of graphs $\langu$ is \cmso-\emph{definable}
(resp. \mso-definable) if there exists a \cmso{} (resp. \mso) sentence
$\phi$ such that $\langu = \set{\graph \mid \strucof{\graph} \models
  \phi}$. For instance, the set of triangle graphs is \mso-definable:

\begin{lemmaE}\label{lemma:triangle-definable}
  $\trianglegraphs$ is an \mso-definable set.
\end{lemmaE}
\begin{proofE}
  By Definition \ref{def:triangle-graph}, $\trianglegraphs$ is the set
  of graphs of tree-width at most $3$ having the triangle
  property. The \mso{} sentence defining $\trianglegraphs$ is the
  conjunction of the following: \begin{compactitem}[-]
  \item $\phi_{\mathrm{tw}\leq3}$ is the sentence defining the set of
    graphs $\set{\graph \mid \twof{\graph} \leq 3}$. This sentence is
    the conjunction of the negation of four \mso{} sentences defining
    the four excluded minors of this set~\cite{ARNBORG19901}.
  \item the sentence:
    \[\phi_{\mathit{triangle}} \isdef \forall e \forall x \forall y
    ~.~ \edgrel(e,x,y) \rightarrow \exists z \exists f \exists g ~.~ z
    \neq x \wedge z \neq y \wedge \edgrel(f,x,z) \wedge
    \edgrel(g,y,z)\] defines the set of graphs having the triangle
    property.
  \end{compactitem}
  \vspace*{-\baselineskip}
\end{proofE}

An important result due to Courcelle is that any \cmso-definable set
of graphs is recognizable in \hr~\cite[Theorem
  4.4]{CourcelleI}\footnote{But not viceversa, i.e., there are
infinitely many recognizable sets that are not \cmso-definable.}. A
consequence of this fact and of Theorem \ref{thm:ft} is the following:

\begin{theorem}[Theorem 1.22 in~\cite{courcelle_engelfriet_2012}]\label{thm:ft-logic}
  Let $\grammar$ be a \hr{} grammar and $\phi$ be a \cmso{}
  sentence. Then, one can build a grammar $\grammar'$ such that
  $\langof{\text{\hr}}{\grammar'} = \langof{\text{\hr}}{\grammar}
  \cap\set{\graph \mid \strucof{\graph} \models \phi}$.
\end{theorem}

The first result of this section is a consequence of the above theorem:

\begin{corollaryE}\label{cor:cmso-decidability}
  Let $\langu$ be a context-free set of triangle (resp. fan) graphs
  and $\phi$ be a \cmso{} sentence. The problem ``\emph{is there a
    graph} $\graph\in\langu$ \emph{such that } $\graph \models
  \phi$?'' is decidable.
\end{corollaryE}
\begin{proofE}
  Let $\grammar$ be a grammar such that $\langu =
  \langof{\trianglealg}{\grammar}$ and $\overline\grammar$ be the
  grammar obtained by replacing each function symbol $f \in
  \trianglesign$ from the right-hand side of a rule in $\grammar$ by
  its definition in \hr. Clearly, we have
  $\langof{\text{\hr}}{\overline\grammar} =
  \langof{\triangle}{\grammar}$. By Theorem \ref{thm:ft-logic}, one
  can build a \hr{} grammar $\grammar'$ such that
  $\langof{\text{\hr}}{\grammar'} =
  \langof{\text{\hr}}{\overline\grammar} \cap \set{\graph \mid
    \strucof{\graph} \models \phi} = \langu \cap \set{\graph \mid
    \strucof{\graph} \models \phi}$. Since the emptiness of a grammar
  is a decidable problem, one obtains the decidability of the problem
  from the statement. The proof is similar for fan graphs.
\end{proofE}

\subsection{Parsable Sets}
\label{subsec:parsable}

This subsection is devoted to the proof of the equivalence between
recognizability and \cmso-definability, for sets of triangle and fan
graphs, respectively:

\begin{theorem}\label{thm:main}
 Let $\langu$ be a set of triangle (resp. fan) graphs. $\langu$ is
 recognizable in the triangle (resp. fan) algebra if and only if
 $\langu$ is \cmso-definable.
\end{theorem}

The usual method for proving such equivalences is to show that, from
each triangle (resp. fan) graph $\graph$, it is possible to define in
\mso{} a triangle (resp. fan) term $t$ such that
$\hval(t)=\graph$. These definitions of terms from graphs rely on the
notion of \emph{transductions}, briefly explained below. To avoid
clutter, the formal definition of transductions is given in Appendix
\ref{app:transductions}.

Let $\relations$ and $\relations'$ be two relational signatures, not
necessarily disjoint. A $(\relations,\relations')$-\emph{transduction}
is a relation $\trans \subseteq \strucof{\relations} \times
\strucof{\relations'}$ defined by the application of the below
operations in this order: \begin{compactenum}
\item \emph{Coloring}. This operation interprets $n$ set variables
  $X_1,\ldots,X_n$ as any subsets $U_1,\ldots,U_n$ of the universe of
  the input $\relations$-structure. The choice of $U_1,\ldots,U_n$ is
  visible throughout the transduction.
\item \emph{Copying}. This operation creates $k$ disjoint copies
  (layers) of the input $\relations$-structure. The elements of the
  copies are identified by binary relations $\copyof{i}(x,y)$ (i.e.,
  $y$ is the copy of $x$ in the $i$-th layer), for all $i=1,\ldots,k$.
\item \emph{Interpreting}. This operation uses the following family of
  \mso{} formul{\ae} written in the $\relations$ signature:
  \begin{align*}
    & ~\domof(X_1,\ldots,X_n) \hspace*{12mm}
    \univof{i}(x,X_1,\ldots,X_n),~ 1 \leq i \leq k \\
    & ~\arel_{i_1,\ldots,i_k}(x_1,\ldots,x_{\arityof{\arel}},X_1,\ldots,X_n),~ \arel\in\relations',~1 \leq i_1,\ldots,i_k \leq n
  \end{align*}
  to define, respectively, the domain of $\trans$ as the set of
  $\relations$-structures that satisfy $\domof(X_1,\ldots,X_n)$ under
  the chosen coloring $U_1,\ldots,U_n$ and, for each input structure $\astruc \in \dom{\trans}$: \begin{compactitem}[-]
  \item the universe of the $i$-th layer in the output
    $\relations'$-structure is the set of elements $u$, taken from the
    $i$-th copy, such that $\astruc \models^\store
    \univof{i}(x,X_1,\ldots,X_n)$ under the assignment $\store(x)=u$
    and $\store(X_i)=U_i$, for all $1\leq i \leq n$. The universe of
    the output is the disjoint union of these individual-layer
    universes.
  \item the interpretation of each relation symbol
    $\arel\in\relations'$ in the output structure is the set of tuples
    $(u_1,\ldots,u_{\arityof{\arel}})$ such that $\astruc
    \models^\store
    \arel_{i_1,\ldots,i_k}(x_1,\ldots,x_{\arityof{\arel}},X_1,\ldots,X_n)$
    under the assignment $\store(x_j)=u_j$, where each $u_j$ is taken
    from the $i_j$-th copy, for all $1 \leq j \leq \arityof{\arel}$
    and $\store(X_i)=U_i$, for all $1 \leq i \leq n$.
  \end{compactitem}
\end{compactenum}
Note that coloring is the only nondeterministic operation: once the
choice of the parameters $X_1,\ldots,X_n$ is made, the transduction
has at most one output, since the input structure or the choice of
parameters may be rejected by the $\domof(X_1,\ldots,X_n)$ formula.
Transductions enjoy the following properties:

\begin{proposition}[Theorem 1.40 in \cite{courcelle_engelfriet_2012}]
  \label{prop:bt}
  If $\langu \subseteq \strucof{\relations'}$ is a \cmso-definable
  (resp. \mso-definable) set and $\trans$ is a
  $(\relations,\relations')$-transduction then the set
  $\trans^{-1}(\langu)$ is \cmso-definable
  (resp. \mso-definable). Moreover, the domain-restriction of a
  transduction by an \mso-definable set and the composition of two
  transductions are transductions.
\end{proposition}
As an application of the first property above (also known as the
\emph{Backwards Translation Theorem}), we obtain that the set of fan
graphs is \mso-definable:

\begin{lemmaE}\label{lemma:fan-definable}
  $\trianglefangraphs$ is an \mso-definable set.
\end{lemmaE}
\begin{proofE}
  By Definition \ref{def:fan-graph}, the \mso{} sentence defining
  $\trianglefangraphs$ is:
  \begin{align*}
    \phi_{\fan{}} \isdef \phi_{\basictriangle} \wedge (\phi_{\text{3-connected}} \vee \overline{\phi}_{\text{3-connected}})
  \end{align*}
  where $\phi_{\basictriangle}$ is the \mso{} sentence defining the
  set of triangle graphs (Lemma \ref{lemma:triangle-definable}) and
  \begin{align*}
    \phi_{\text{3-connected}} \isdef
    & ~\underbrace{\exists X ~.~ \mathit{isVert}(X) \wedge \cardof{X} \geq 3}_{\text{at least } 3 \text{ vertices}}
    ~\wedge \\
    & ~\underbrace{\forall x \forall y \forall X ~.~ \mathit{isVert}(X) \wedge \neg(\cardof{X} \geq 3)
      \rightarrow \mathit{path\_avoid}(x,y,X)}_{\text{cannot be disconnected by removing $<3$ vertices}}
  \end{align*}
  where $\mathit{isVert}(X)$ means that $X$ is a set of vertices only,
  $\cardof{X} \geq k$ means that there are at least $k$ distinct
  elements in $X$ and $\mathit{path\_avoid}(x,y,X)$ means that there
  is a path between the vertices $x$ and $y$ that avoids the vertices
  in $X$. The formal definitions of the latter formul{\ae} are
  standard, hence we omit them.

  The definition of $\overline{\phi}_{\text{3-connected}}$ is
  explained below. Let $\trans$ be the transduction taking as input a
  $3$-graph and adding one edge between each of its sources, i.e.,
  this transduction computes the thickening of a $3$-graph. Clearly,
  this transformation can be defined using \mso{} formul{\ae}. Then
  $\overline{\phi}_{\text{3-connected}}$ is the \mso{} sentence
  defining the set $\trans^{-1}(\langu_{\text{3-connected}})$, where
  $\langu_{\text{3-connected}}$ is the set defined by
  $\phi_{\text{3-connected}}$ above. The existence of
  $\overline{\phi}_{\text{3-connected}}$ is guaranteed by Proposition
  \ref{prop:bt}.
\end{proofE}

We shall define transductions building terms from graphs. Note that
terms are ranked and ordered trees whose nodes are labeled by function
symbols from a signature $\signature$, such that a $f$-labeled node
has exactly $\arityof{f}$ ordered children. A notable relaxation of
this constraint is the representation of the nodes labeled by parallel
composition ($\pop$): since, in all the graph algebras considered
here, the $\pop$ operation is both associative and commutative, all
$\pop$-labeled nodes of a term are merged into one node, whose
children may occur in any order. For instance, the term $(x \pop y)
\pop z$ has one root node whose children are $x$, $y$ and $z$ and
their order is not important. Note that, because we aim at generating
terms via transductions, chosing such an ordering would mean chosing a
relation\footnote{Quantifying existentially over relations is possible
in Second Order logic but not in \mso.}, which is beyond the
expressiveness of \mso.

\begin{figure}[t!]
  \centerline{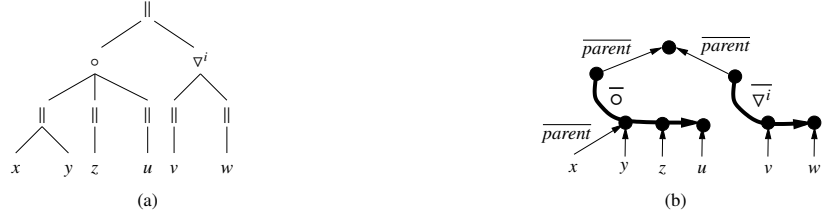}
  \caption{The structure (b) encodes the fan term $\sop(x \pop y, z,
    u) \pop v \fan{i} w$, for some $i \in \set{1,2,3}$ (a). The arrows
    indicate the order in the tuples of the structure encoding the
    term (b). }
  \label{fig:term}
\end{figure}

Formally, let $\overline{\signature} \isdef \set{\overline{f} \mid f
  \in \signature \setminus \set{~\pop~}} \cup \set{\parentof}$ be a
signature of relation symbols, where $\arityof{\overline{f}} =
\arityof{f}+1$, for each $f \in \signature\setminus\set{~\pop~}$ and
$\parentof$ is binary. To simplify the encoding, we assume that the
children of a term node labeled with a symbol from
$\signature\setminus\set{~\pop~}$ are all labeled by $\pop$. In other
words, we consider that the parallel composition does not have a fixed
arity, but instead can be applied to one or more arguments, where
$\pop \graph \isdef \graph$, for each graph $\graph$, by
convention. In this way, non-$\pop$-nodes alternate with $\pop$-nodes
on each path in the tree representing a term, as in Figure
\ref{fig:term} (a).

The \emph{tree-encoding} of a ground $\signature$-term $t$ is the
$\overline{\signature}$-structure $\strucof{t}$, whose universe is the
set of positions of subterm occurrences in $t$ and the relation
symbols from $\overline{\signature}$ are interpreted as
follows: \begin{compactitem}[-]
\item for each $f \in \signature\setminus\set{~\pop~}$, the relation
  $\overline{f}^{\strucof{t}}$ consists of the tuples
  $(n_0,n_1,\ldots,n_{\arityof{f}})$, where each $n_0$ is a
  $f$-labeled node of $t$ and $n_1,\ldots,n_{\arityof{f}}$ are its
  children, in this order.
\item $\parentof^{\strucof{t}}$ is the set of pairs $(m,n)$, where $m$
  is a $\pop$-labeled node of $t$ and $n$ is a child of $m$.
\end{compactitem}
For instance, Figure \ref{fig:term} (b) shows the tree-encoding of the
fan term in Figure \ref{fig:term} (a). Note that each tree-encoding of
a ground $\signature$-term encodes a unique such term (modulo
the associativity and commutativity of the parallel composition). The
encoding of terms as structures mimicks closely the notion of
\emph{reduced trees}\footnote{The main difference is that reduced tree
do not have parent edges as our tree-encodings.} of
Courcelle~\cite[Definition 3.2]{CourcelleV}. This underlies the
definition of \emph{parsable} sets of graphs, for a given algebra
$\algof{A}$ of graphs, over the signature $\signature^\algof{A}$ and
having domain $\universeOf{A}$ (see also~\cite[Definition 4.2]{CourcelleV}):

\begin{definition}\label{def:parsable}
  A subset $\mathcal{C} \subseteq \universeOf{A}$ is \emph{parsable in
    $\algof{A}$} if and only if there exists a transduction $\trans$
  between graphs and terms such that, for each graph $\graph \in
  \mathcal{C}$, each term-encoding $\astruc \in
  \trans(\strucof{\graph})$ encodes a (unique) ground
  $\signature^\algof{A}$-term $t$ such that $\hval(t)=\graph$.
\end{definition}
Parsable sets are quintessential in proving the equivalence between
recognizability and \cmso-definability, as stated by the following
theorem, due to Courcelle:

\begin{theorem}[Theorem 4.8 (2) in~\cite{CourcelleV}]\label{thm:rec-def}
  Let $\mathcal{C} \subseteq \universeOf{A}$ be a parsable set in
  $\algof{A}$. For each subset $\langu \subseteq \mathcal{C}$,
  $\langu$ is recognizable in $\algof{A}$ if and only if $\langu$ is
  \cmso-definable.
\end{theorem}

In the following, we shall apply the above theorem to the triangle and
fan algebras, respectively. To show that these sets of graphs are
parsable, we rely on a seminal result of Boja\'{n}czyk and Pilipczuk,
proving the existence of a transduction producing an optimal tree
decomposition of a graph of tree-width bounded by a given
constant~\cite{10.1145/2933575.2934508,journals/lmcs/BojanczykP22}. Before
stating this result, we recall their encoding of tree decompositions
using structures, over the following signature of relation
symbols: \begin{compactitem}[-]
\item $\node(x)$ means that $x$ is a node of the tree decomposition,
\item $\bagof(x,y)$ means that the vertex $x$ is in the bag of the node
  $y$,
\item $\parent(x,y)$ means that the node $x$ is the parent of the node
  $y$ in the tree decomposition,
\item $\edgrel(x,y,z)$ describes the incidence relation in the graph
  (\autoref{subsec:cmso})
\end{compactitem}
Importantly, the proof of the below theorem relies on the fact that
input graphs are encoded using the incidence model
(\autoref{subsec:cmso}). We adopt the same encoding of graphs by
structures, in order to use this seminal result in our proofs:

\begin{theorem}[Corollary 2.2 in \cite{journals/lmcs/BojanczykP22}]
  \label{thm:td-trans}
  For each $k \in \nat$, there exists a transduction $\mathcal{I}$
  from graphs to tree decompositions such that, for every input graph
  $\graph$: \begin{compactenum}
  \item if $\twof{\graph} \leq k$ then $\mathcal{I}(\graph) \neq
    \emptyset$, and
  \item each structure from $\mathcal{I}(\graph)$ is the encoding of an
    optimal tree decomposition of $\graph$.
  \end{compactenum}
\end{theorem}
\noindent An immediate consequence of this result is that the set of
graphs $\set{\graph \mid \twof{\graph} \leq k}$ is parsable in \hr,
for any given $k \in \nat$~\cite[Theorem
  6.8]{journals/corr/abs-2310-04764}.  Note that the transduction
$\mathcal{I}$ from the above theorem is guaranteed to return an
optimal tree decomposition, with no further assumption on the shape of
the decomposition. We make sure that Assumption \ref{ass:td} holds on
each output of this transduction, by composing $\mathcal{I}$ with a
transduction between tree decompositions that contracts the edges $mn
\in \edgeof{\tree}$ such that $\bag(m) \subseteq \bag(n)$ in the input
tree decomposition $(\tree,\bag)$. Such a transduction exists because,
in general, the quotienting of a structure by an \mso-definable
equivalence relation is a transduction~\cite[Lemma 2.4]{CourcelleV}.

\subsection{Parsability of the Triangle and Fan Graphs}

We prove next that the set of triangle graphs is parsable in the
triangle algebra. The transduction that builds, from each triangle
graph $\graph$, triangle terms $t$ such that $\hval(t)=\graph$, is
obtained by composing the transduction $\mathcal{I}$ from Theorem
\ref{thm:td-trans} with the transductions introduced by the following
lemmas.

First, we transform each tree decomposition $(\tree,\bag)$ of width
$k$ into an alternating tree decomposition having the $k$-step
property:

\begin{lemmaE}\label{lemma:k-step-atd}
  There exists a transduction $\mathcal{J}$ between tree
  decompositions such that, for each input tree decomposition
  $(\tree,\bag)$ of width $k$, the set $\mathcal{J}(\tree,\bag)$ is
  non-empty and consists of alternating tree decompositions of width
  $k$ having the $k$-step property.
\end{lemmaE}
\begin{proofE}
The construction is obtained as the composition of four
MSO-transductions. Starting from the Bojańczyk-Pilipczuk encoding of
an optimal tree decomposition of width $k$, the first transduction
``floods'' the decomposition by adding vertices to bags until every
bag has size exactly $k+1$.  The second transduction contracts
adjacent nodes having identical bags. The third transduction refines
the decomposition by replacing every edge whose endpoint bags differ
in more than one vertex by a path of intermediate bags, so that
adjacent bags differ by exactly one vertex. Finally, the fourth
transduction introduces explicit adhesion nodes between adjacent bags.
The resulting decomposition is alternating, with bag nodes of size
$k+1$ and adhesion nodes of size $k$, i.e., it has the $k$-step
property.

\vspace*{-\baselineskip}
\paragraph*{First transduction.} We now detail the first step:
\begin{claim}
  There exists an MSO-transduction which, given a graph $\graph$
  together with an encoding of an optimal tree decomposition
  $(\tree,\bag)$ of width $k$, outputs a tree decomposition
  $(\tree,\bag')$ of the same graph and of the same width such that
  every bag contains exactly $k+1$ vertices of $\graph$.
\end{claim}
\begin{proof}
Let $k$ be fixed. We assume that the input already contains an
encoding of the tree decomposition over the signature used by
Boja\'nczyk and Pilipczuk, i.e., encoded by the relations $\node(x)$,
$\bagof(x,y)$, $\parent(x,y)$ and $\edgrel(x,y,z)$, as described
earlier.

Since the input decomposition is optimal of width $k$, there exists a
node $c\in \vertof{\tree}$ with $\cardof{\bag(c)}=k+1$. The
transduction guesses such a node $c$ (using a parameter $C$), which we
will call the \emph{center} in the following.  It also guesses
parameter sets $H_1,\dots,H_{k+1}\subseteq \vertof{\graph}$
interpreted as a $k$-history that maps vertices into the set
$\set{1,\ldots,k+1}$ (Definition \ref{def:history}). The domain
formula of the transduction requires that these predicates form a
partition of $\vertof{\graph}$, and that the history is injective on
every bag, i.e. \[\forall t \in \vertof{\tree}\ \forall x,y\in
\vertof{\graph} ~.~ \bagof(x,t)\wedge \bagof(y,t)\wedge H_i(x)\wedge
H_i(y) \rightarrow x=y\] for every $i\in\{1,\dots,k+1\}$. Note that
such a history always exists (Lemma \ref{lemma:history}). In the
following we refer to the history labeling of vertices defined by
$H_1, \ldots, H_{k+1}$ as to the coloring phase of the transduction.

We now define the new bag relation $\bagof'(x,t)$. For each node $t$
and each color $i$, consider the unique path from $t$ to the center
$c$. Since $\bag(c)$ contains all $k+1$ colors, there exists a first
node $s$ on this path whose bag contains a vertex of color $i$.
Because the coloring is injective on bags, there is a unique such
vertex; denote it by $v_i(t)$. We define
\[
\bag'(t)=\{v_i(t): i=1,\dots,k+1\}.
\]
Equivalently, $\bagof'(v,t)$ holds iff there exists a color
`$i\in\{1,\dots,k+1\}$ and a node $s\in \vertof{\tree}$ such that:
\begin{compactenum}
\item $v \in H_i$,
\item $s$ lies on the path from $t$ to $c$,
\item $\bagof(v,s)$ holds,
\item no node strictly between $t$ and $s$ contains a vertex of color $i$.
\end{compactenum}
All these conditions are \mso-definable over the encoding of the input
tree decomposition. In particular, the existence of a path between two
nodes is \mso-definable in trees.  It remains to prove the correctness
of the above definition.

First, every bag $\bag'(t)$ has size exactly $k+1$.  For every color
$i$, the first occurrence of color $i$ on the path from $t$ to $c$
exists because $\bag(c)$ contains all colors. Moreover, by injectivity
of the coloring on bags, this occurrence determines a unique vertex.
Hence $\bag'(t)$ contains exactly one vertex of each color.

Second, the original bag is contained in the new bag. Indeed, if
$v\in\bag(t)$ contains a vertex of color $i$, then $t$ itself is
already the first node on the path from $t$ to $c$ containing color
$i$.  Hence $v=v_i(t)$, and therefore $v\in\bag'(t)$. Consequently,
\[
\bag(t)\subseteq\bag'(t)
\]
for every node $t$. In particular, every edge that was covered before
is still covered, thus taking care of Definition \ref{def:tw}
(\ref{it2:def:tw}). Moreover, let $n$ be a node such that $\bag(n)$
contains all sources of $\graph$. Then $\bag'(n)$ also contains the
sources of $\graph$, thus taking case of Definition \ref{def:tw}
(\ref{it1:def:tw}). It remains to verify the connectedness condition,
namely Definition \ref{def:tw} (\ref{it3:def:tw}). Fix a vertex $v$,
and let
\[
\tree_v=\{t\in \vertof{\tree}: \bagof(v,t)\},
\qquad
\tree'_v=\{t\in \vertof{\tree}: \bagof'(v,t)\}.
\]

We know that $\tree_v$ is nonempty and connected.
We show that every newly added occurrence of $v$ is attached to $\tree_v$ along a path.
Suppose $v\in\bag'(t)$, and let $i$ be the color of $v$.
By definition, there exists a first node $s$ on the path from $t$ to $c$ such that $v\in\bag(s)$.
Hence $s\in \tree_v$.

Moreover, for every node $r$ on the path from $t$ to $s$, the same node $s$ is still the first node on the path from $r$ to $c$ containing color $i$.
Therefore $v\in\bag'(r)$.
Consequently, the entire path from $t$ to $s$ belongs to $\tree'_v$.
Since $s\in \tree_v$, every newly added occurrence of $v$ is connected to the original connected subtree $\tree_v$.

Therefore $\tree'_v$ is connected for every vertex $v$.
Hence $(\tree,\bag')$ is a valid tree decomposition of $\graph$.
Since every bag has size exactly $k+1$, its width is still $k$.
\end{proof}

\vspace*{-\baselineskip}
\paragraph*{Second transduction.}
The first transduction may create chains of neighbouring nodes carrying exactly the same bag. 
We eliminate nodes with the same bags, quotienting the decomposition tree by the equivalence relation:
\[
t \sim t' \iff \bag'(t)=\bag'(t').
\]
Since bag equality is \mso-definable, so is the relation $\sim$.  By
the standard result that quotienting by an \mso-definable equivalence
relation is an MSO transduction~\cite[Lemma 2.4]{CourcelleV}, this
yields a new tree decomposition of the same graph. Moreover, all bags
still have cardinality $k+1$, and adjacent nodes carry distinct bags.

\vspace*{-\baselineskip}
\paragraph*{Third transduction.}
We now refine the decomposition so that adjacent bags differ by exactly one vertex.
Let $tt'$ be an edge of the decomposition tree and define:
\[
D_t=\bag(t)\setminus\bag(t'),
\qquad
D_{t'}=\bag(t')\setminus\bag(t).
\]
Since all bags have cardinality $k+1$, these sets have the same
cardinality $m$. Moreover, since adjacent bags are distinct, we have
$m\geq 1$.

If $m=1$, we leave the edge unchanged. Otherwise, let:
\[
D_t=\{a_1,\ldots,a_m\},
\qquad
D_{t'}=\{b_1,\ldots,b_m\},
\]
where the vertices are ordered by increasing color. The transduction
replaces the edge $tt'$ by a path whose intermediate bags are:
\[
(\bag(t)\setminus\{a_1,\ldots,a_i\})
\cup
\{b_1,\ldots,b_i\},
\qquad
1\leq i<m.
\]
Since the coloring from the first transduction is injective on every
bag, these orderings are uniquely determined.  Moreover, $m\leq k+1$,
and hence at most $k$ intermediate nodes are required.  As $k$ is
fixed, the transduction can realize this construction by creating $k$
copies of every tree edge and using \mso{} formulae to determine which
copies are present and which bag they carry.

The resulting decomposition still represents the same graph. Indeed,
every original bag is preserved, so the coverage of sources and edges
remains unchanged. Moreover, for every graph vertex $v$, the set of
nodes whose bags contain $v$ remains connected: along each newly
introduced path, the occurrences of $v$ form a connected
subpath. Consequently, the tree-decomposition property is preserved.
After this step, every bag still has size $k+1$, and for every
adjacent pair of nodes $t,t'$ we have
$\cardof{\bag(t)\setminus\bag(t')} = \cardof{\bag(t')\setminus\bag(t)}
= 1$.

\vspace*{-\baselineskip}
\paragraph*{Fourth transduction.}
Let $r$ be a bag containing the set of sources. Such a bag exists by
the definition of tree decompositions. We introduce a fresh root
adhesion node whose bag consists precisely of the sources and connect
it to $r$.  We then apply the construction from the proof of
Lemma~\ref{lemma:alternating-td}, inserting an adhesion node on every
edge of the decomposition tree and labeling it with the intersection
of the bags of its two neighbouring bag nodes.  Since the new adhesion
nodes correspond bijectively to edges of the decomposition tree and
their bags are definable as intersections of adjacent bags, this
construction is \mso-definable. Since adjacent bag nodes differ by
exactly one vertex and have cardinality $k+1$, every adhesion bag has
cardinality $k$. The resulting decomposition is alternating and has
the $k$-step property.
\end{proofE}

Second, each alternating tree decomposition of width $k$ having the
$k$-step property is transformed into the tree-encoding of a triangle
term, by the following transduction:

\begin{lemmaE}\label{lemma:atd-triangle-term}
  There exists a transduction $\mathcal{K}^\trianglealg$ between tree
  decompositions and triangle terms such that, for each input tree
  decomposition $(\tree,\bag)$ of a graph $\graph$, if $\graph$ is
  triangle graph and $(\tree,\bag)$ is an alternating tree
  decomposition of $\graph$ having the $3$-step property, then the set
  $\mathcal{K}^\trianglealg(\tree,\bag)$ is non-empty and consists of
  tree-encodings of triangle terms $t$ such that $\hval(t)=\graph$.
\end{lemmaE}
\begin{proofE}
  We describe the three phases of $\mathcal{K}^\trianglealg$. For the
  interpretation phase, we sketch the construction of the respective
  \mso{} formul{\ae}.

  \vspace*{.5\baselineskip}\noindent\emph{Coloring}. The parameters of
  the transduction are $H_1, H_2, H_3$ and $H_4$, ranging over the
  vertices of $\graph$, for the choice of the $3$-history and $A,B$,
  ranging over the nodes of $\tree$, for the sets of adhesion and bag
  nodes.

  \vspace*{.5\baselineskip}\noindent\emph{Copying}. The transduction
  copies the input into $5$ layers. The nodes from the first layer are
  used to represent the nodes of the output term in the structure that
  encodes it, whereas the layers $2$, $3$ and $4$ contain copies of
  the nodes used to represent the additional $\basictriangle$ and
  $\emptygraph$ constants that may occur in this term. As a
  convention, a node from the $i$-th layer occurs on the $i$-th
  position in a tuple from the output structure, for $1 \leq i \leq
  4$. The $5^\mathit{th}$ layer is used for the children of adhesion
  nodes used to represent the $\basictriangle$ and $\emptygraph$
  constants joined in parallel to these terms. The construction is
  illustrated in Figure \ref{fig:sop} (rightmost column), where the
  notation ($n,i$) means that the node $n$ is taken from the $i$-th
  layer.

  \vspace*{.5\baselineskip}\noindent\emph{Interpreting}. The domain,
  universe and relations of the output are defined by the following
  \mso{} formul{\ae}, respectively: \begin{compactitem}[-]
  \item $\domof(A,B,H_1,H_2,H_3,H_4)$ is the conjunction of \mso{}
    formul{\ae} stating the following facts: \begin{compactenum}
    \item $\graph$ is a triangle graph: there exists an \mso{}
      sentence defining the set $\trianglegraphs$, by Lemma
      \ref{lemma:triangle-definable}. This sentence can be used here
      as such, because the vertices and edges of $\graph$ are
      available in the input structure, by the definition of the
      encoding of tree decompositions as structures.
    \item $(\tree,\bag)$ is $(A,B)$-bipartite, such that the root of
      $\tree$ belongs to $A$, its leaves belong to $B$ and each path
      from the root to a leaf alternates between nodes in $A$ and
      nodes in $B$. Moreover, all nodes in $A$ (resp. $B$) have bags
      of cardinality $3$ (resp. $4$).
    \item $H_1$, $H_2$, $H_3$ and $H_4$ meet the conditions from
      Definition \ref{def:history}, in particular no two vertices from
      a bag belong to the same set.
    \end{compactenum}
  \item $\univof{i}(x,A,B,H_1,H_2,H_3,H_4) \isdef \node(x)$, for all
    $1\leq i\leq4$, i.e., the universe of each layer is the set of
    nodes of $\tree$.
  \item $\parentof_{i,j}(x,y,A,B,H_1,H_2,H_3,H_4)$ is
    either: \begin{compactitem}[*]
  \item $A(x) \wedge B(y) \wedge \parent(x,y)$ if $(i,j)=(1,1)$, i.e.,
    the parent is an adhesion node and the child is a bag node, both
    taken from the first layer, or
  \item $A(x) \wedge x=y$ if $(i,j)=(1,5)$, i.e., both the parent and
    child are copies of the same adhesion node taken from the
    $1^\mathit{st}$ and $5^\mathit{th}$ layers, respectively.
  \item $\false$ if $(i,j) \not\in \set{(1,1),(1,5)}$.
  \end{compactitem}
  \item $\overline{\sop}_{i,j,k,\ell}(x,y,z,s,A,B,H_1,H_2,H_3,H_4)$ is
    defined according to the base and inductive cases of the proof of
    Theorem \ref{thm:triangle-graphs}, respectively. In both cases,
    $x$ is a bag node, i.e., $B(x)$ holds. \begin{compactitem}[*]
    \item (\emph{Base case}) If $x$ is a leaf of $\tree$ then $y$, $z$
      and $s$ are the copies of the parent of $x$ in $\tree$ (i.e.,
      the unique adhesion node $m$ such that $\parent(m,x)$ holds)
      taken from the $2^\mathit{nd}$, $3^\mathit{rd}$ and
      $4^\mathit{th}$ layers, respectively, i.e.,
      $(i,j,k,\ell)=(1,2,3,4)$. See Figure \ref{fig:sop} (a) for a
      depiction of the output structure in this case. Moreover,
      $\overline{\sop}_{i,j,k,\ell} \isdef \false$ if $(i,j,k,\ell)
      \neq (1,2,3,4)$.
    \item (\emph{Inductive case}) Else, $x$ is a non-leaf bag node of
      $\tree$ with adhesion parent $m$, having at most three adhesion
      children $p_1$, $p_2$ and $p_3$. Then, $y$, $z$ and $s$ are the
      copies of $p_1$, $p_2$ and $p_3$ from the first layer, i.e.,
      $(i,j,k,l)=(1,1,1,1)$.  If some child $p_s$ is missing it is
      replaced by the copy of $m$ from the $(s+1)$-th layer, for $1
      \leq s \leq 3$, as in the previous case, i.e., $j=s+1$ or
      $k=s+1$ or $\ell=s+1$. See Figure \ref{fig:sop} (b) for a
      depiction of the output structure in this case. Moreover,
      $\overline{\sop}_{i,j,k,\ell}(x,y,z,s,A,B,H_1,H_2,H_3,H_4)
      \isdef \false$ if $i\neq1$ or in any of the missing cases.
    \end{compactitem}
  \item $\overline{\basictriangle}_i(x,A,B,H_1,H_2,H_3,H_4)$ and
    $\overline{\emptygraph}_i(x,A,B,H_1,H_2,H_3,H_4)$, for the layers
    $2 \leq i \leq 4$, are defined according to the same case split as
    for the previous point. In each case, $x$ is an adhesion node
    (i.e., $A(x)$ holds) such that the bag of $x$ consists of vertices
    $u$, $v$ and $w$. Taking into account the edges of $\graph$, we
    distinguish whether $\overline{\basictriangle}_i(x)$ holds (i.e.,
    $uvw$ is a subgraph of $\graph$) or $\overline{\emptygraph}_i(x)$
    holds (i.e., $uvw$ is not a subgraph of $\graph$).
  \end{compactitem}
  The proof of the fact that $\mathcal{K}^\trianglealg(\tree,\bag)$ is
  non-empty and consists of tree-encodings of terms $t$ such that
  $\hval(t) = \graph$ follows from the argument used in the proof of
  Theorem \ref{thm:triangle-graphs}.
\end{proofE}

\vspace*{-\baselineskip}
\paragraph*{Proof of Theorem \ref{thm:main} (the triangle case)}
By composing the transductions $\mathcal{I}$, $\mathcal{J}$ and
$\mathcal{K}^\trianglealg$ and taking the domain-restriction to
$\trianglegraphs$ (the set $\trianglegraphs$ is \mso-definable, by
Lemma \ref{lemma:triangle-definable}), we obtain a transduction from
triangle graphs to triangle terms that outputs one or more terms
encoding the input (triangle) graph (Proposition \ref{prop:bt}). By
Definition \ref{def:parsable}, the set $\trianglegraphs$ is parsable
in the triangle algebra, hence recognizability and \cmso-definability
coincide over triangle graphs, by Theorem \ref{thm:rec-def}. \qed

In the case of fan graphs, the proof of Theorem \ref{thm:main} uses
the following counterpart of Lemma \ref{lemma:atd-triangle-term}:

\begin{lemmaE}\label{lemma:atd-fan-term}
  There exists a transduction $\mathcal{K}^\fanalg$ between tree
  decompositions and fan terms such that, for each input tree
  decomposition $(\tree,\bag)$ of a graph $\graph$, if $\graph$ is fan
  graph and $(\tree,\bag)$ is an alternating tree decomposition of
  $\graph$ having the $3$-step property, then the set
  $\mathcal{K}^\fanalg(\tree,\bag)$ is non-empty and consists of
  tree-encodings of fan terms $t$ such that $\hval(t)=\graph$.
\end{lemmaE}
\begin{proofE}
  The definition of $\mathcal{K}^\fanalg$ is similar to that of
  $\mathcal{K}^\trianglealg$ from the proof of Lemma
  \ref{lemma:atd-triangle-term}. In particular, the coloring and
  copying phases are the same (except that $\mathcal{K}^\fanalg$ only
  requires $4$ layers). We point out the differences that occur in the
  interpretation of the following relations: \begin{compactitem}[-]
  \item $\overline{\sop}_{i,j,k,\ell}(x,y,z,s,A,B,H_1,H_2,H_3,H_4)$ is
    defined according to the inductive case of the proof of Theorem
    \ref{thm:triangle-graphs}, where $x$ denotes a bag node taken from
    the first layer, i.e., $B(x)$ holds and $i=1$ (i.e.,
    $\overline{\sop}_{i,j,k,\ell}(x,y,z,s,A,B,H_1,H_2,H_3,H_4) \isdef
    \false$ if $i \neq 1$). We distinguish the following
    cases: \begin{compactitem}[*]
    \item $x$ has a single adhesion child taken from the
      $1^\mathit{st}$ layer, either $y$, $z$ or $s$ (i.e., $j=1$,
      $k=1$ or $\ell=1$, respectively) and the two other subsets of
      cardinality $3$ from the bag of $x$ (i.e., except for the
      adhesion of $x$ to its parent and the adhesion of the child of
      $x$) induce $K_3$ subgraphs. The missing children are replaced
      by a copy of $x$ taken from the $2^\mathit{nd}$, $3^\mathit{rd}$
      or $4^\mathit{th}$ layers, according to each case. The order of
      the arguments is determined by $H_1$, $H_2$, $H_3$ and
      $H_4$. See Figure \ref{fig:fan} (c) for a depiction of this
      case, where the child of $x$ is $y$ and $H_1(x)$, $H_2(y)$,
      $H_3(z)$ and $H_4(a)$ hold.
    \item $x$ has two adhesion children taken from the $1^\mathit{st}$
      layer, either $y$, $z$ or $s$ (i.e., two of $j$, $k$ and $\ell$
      equal $1$, respectively) and one subset of cardinality $3$ of
      the bag of $x$ induce a $K_3$ subgraph. The missing child is
      replaced by a copy of $x$ taken from the $2^\mathit{nd}$,
      $3^\mathit{rd}$ or $4^\mathit{th}$ layers, according to each
      case. The order of the arguments is determined by $H_1$, $H_2$, 
       $H_3$ and $H_4$. See Figure \ref{fig:fan} (e) for a depiction of this
      case, where the children of $x$ are $y$ and $s$, and $H_1(x)$,
      $H_2(y)$, $H_3(z)$ and $H_4(a)$ hold.
    \item $x$ has three adhesion children taken from the
      $1^\mathit{st}$ layer, namely $x$, $y$ and $z$ (i.e.,
      $j=k=\ell=1$). The order of the arguments is determined by
      $H_1$, $H_2$, $H_3$ and $H_4$. See Figure \ref{fig:fan} (f) for
      a depiction of this case, where $H_1(x)$, $H_2(y)$, $H_3(z)$ and
      $H_4(a)$ hold.
    \end{compactitem}
  \item $\overline{\fan{i}}_{j,k,\ell}(x,y,z,A,B,H_1,H_2,H_3,H_4)$ is
    defined according to the inductive case of the proof of Theorem
    \ref{thm:triangle-graphs}, where $x$ denotes a bag node taken from
    the first layer, i.e., $B(x)$ holds and $j=1$ (i.e.,
    $\overline{\fan{i}}_{j,k,\ell}(x,y,z,A,B,H_1,H_2,H_3,H_4) \isdef
    \false$ if $j \neq 1$). We distinguish the following
    cases: \begin{compactitem}[*]
    \item $x$ is a leaf and $y=z$ is its parent, both taken from the
      $1^\mathit{st}$ layer (i.e., $k=\ell=1$). Moreover,
      $\overline{\fan{i}}_{j,k,\ell}(x,y,z,A,B,H_1,H_2,H_3,H_4) \isdef \false$ if
      $(k,\ell)\neq(1,1)$. Then, $1 \leq i \leq 3$ and the order
      of the arguments is determined by $H_1$, $H_2$, $H_3$ and $H_4$. See
      Figure \ref{fig:fan} (a) for a depiction of this case, where
      $i=2$ and $H_1(x)$, $H_2(y)$, $H_3(z)$ and $H_4(a)$ hold.
    \item $x$ has one adhesion child, either $y$ or $z$, taken from
      the $1^\mathit{st}$ layer (i.e., $k=1$ or $\ell=1$,
      respectively) and the other subset of cardinality $3$ from the
      bag of $x$ (i.e., except for the adhesion of $x$ to its parent
      and the adhesion of the child of $x$) induces a $K_3$
      subgraph. The missing child is replaced by a copy of $x$ taken
      from the $2^\mathit{nd}$ or $3^\mathit{rd}$ layers, according to
      each case. Then, $1 \leq i \leq 3$ and the order of the
      arguments is determined by $H_1$, $H_2$, $H_3$ and $H_4$. See
      Figure \ref{fig:fan} (b) for a depiction of this case, where
      $i=2$ and $H_1(x)$, $H_2(y)$, $H_3(z)$ and $H_4(a)$ hold.
    \item $x$ has two adhesion children, namely $y$ and $z$, both
      taken from the $1^\mathit{st}$ layer (i.e., $k=\ell=1$). Then,
      $1 \leq i \leq 3$ and the order of the arguments is determined
      by $H_1$, $H_2$, $H_3$ and $H_4$. See Figure \ref{fig:fan} (d)
      for a depiction of this case, where $i=2$ and $H_1(x)$,
      $H_2(y)$, $H_3(z)$ and $H_4(a)$ hold.
    \end{compactitem}
  \item $\overline{\basictriangle}_i(x,A,B,H_1,H_2,H_3,H_4)$, for the
    layers $2 \leq i \leq 4$, are defined according to the same case
    split as for the previous point. In each case, $x$ is an adhesion
    node (i.e., $A(x)$ holds) such that the bag of $x$ consists of
    vertices $u$, $v$ and $w$, such that $uvw$ is a subgraph of
    $\graph$. The latter condition can be checked taking into account
    the edges of $\graph$ from the encoding of $(\tree,\bag)$.
  \end{compactitem}
  The proof of the fact that $\mathcal{K}^\fanalg(\tree,\bag)$ is
  non-empty and consists of tree-encodings of terms $t$ such that
  $\hval(t) = \graph$ follows from the argument used in the proof of
  Theorem \ref{thm:fan-graphs}.
\end{proofE}

\vspace*{-\baselineskip}
\paragraph*{Proof of Theorem \ref{thm:main} (the fan case)}
By Definition \ref{def:fan-graph}, a $3$-graph $\graph$ is a fan graph
if and only if either it is $3$-connected, or its thickening
$\overline{\graph} = \graph \pop \basictriangle$ is $3$-connected.
The two cases can be identified by the \mso{} sentence that defines
the set of $3$-connected graphs, as in the proof of Lemma
\ref{lemma:fan-definable}.

In the first case, we apply to the input graph the transduction
$\mathcal{I}$ from Theorem \ref{thm:td-trans} composed with a
transduction $\mathcal{J}$ that outputs an alternating tree
decomposition of $\graph$. Such a transduction is defined as in the
proof of Lemma \ref{lemma:k-step-atd}. By Lemma
\ref{lemma:connectivity-tw}, this alternating tree decomposition has
the $3$-step property. Then, $\mathcal{I} \circ \mathcal{J}$ is
composed with the $\mathcal{K}^{\fan{}}$ transduction from Lemma
\ref{lemma:atd-fan-term}, whose output consists of tree-encodings of
fan terms $t$ such that $\graph=\hval(t)$.

In the second case, a first transduction computes the thickening
$\overline{\graph} = \graph \pop \basictriangle$ of $\graph$. Then, we
apply $\mathcal{I} \circ \mathcal{J} \circ \mathcal{K}^{\fan{}}$ to
$\overline{\graph}$, as in the first case above. The final
transduction removes the $\basictriangle$ child of the adhesion root,
thus yielding the tree-encoding of a term $t$ such that
$\graph=\hval(t)$. By Definition \ref{def:parsable}, the set
$\trianglefangraphs$ is parsable in the fan algebra, hence
recognizability and \cmso-definability coincide over fan graphs, by
Theorem \ref{thm:rec-def}. \qed

\section{Conclusions and Future Work}

We have identified a class of graphs characterized by the fact that
edges belong to cliques of three and the tree-width is at most
three. A subclass is further defined by a 3-connectivity
condition. First, for each of these classes of graphs, we provide
characterizations via algebras that extend the well-known
series-parallel graphs of tree-width two. The decidability of Monadic
Second Order with modulo constraints on the cardinality of sets
(\cmso) is a consequence of these algebraic characterization results.
Second, we prove that definability is equivalent to recognizability in
each of these algebras.

As future work, we plan to investigate applications of this new theory
to the decomposition of 3-connected graphs. Moreover, we aim at
defining syntactic descriptions of the recognizable sets, such as
regular grammars or regular expressions and investigate the
algorithmic complexity of their decision problems, such as emptiness,
membership and inclusion. Particularly interesting applications of
these results include the areas of verification, synthesis and
learning.

\bibliographystyle{plain}
\bibliography{refs}

@inproceedings{DBLP:conf/soda/KurkofkaP26,
  author       = {Jan Kurkofka and
                  Tim Planken},
  editor       = {Kasper Green Larsen and
                  Barna Saha},
  title        = {A Tutte-type canonical decomposition of 3- and 4-connected graphs},
  booktitle    = {Proceedings of the 2026 Annual {ACM-SIAM} Symposium on Discrete Algorithms,
                  {SODA} 2026, Vancouver, BC, Canada, January 11-14, 2026},
  pages        = {2942--3021},
  publisher    = {{SIAM}},
  year         = {2026},
  url          = {https://doi.org/10.1137/1.9781611978971.110},
  doi          = {10.1137/1.9781611978971.110},
  bibsource    = {dblp computer science bibliography, https://dblp.org}
}

@article{Cunningham_Edmonds_1980,
title={A Combinatorial Decomposition Theory},
volume={32},
DOI={10.4153/CJM-1980-057-7},
number={3},
journal={Canadian Journal of Mathematics},
author={Cunningham, William H. and Edmonds, Jack},
year={1980},
pages={734–765}}

@InProceedings{10.1007/BFb0017382,
author="Arnborg, Stefan
and Courcelle, Bruno
and Proskurowski, Andrzej
and Seese, Detlef",
editor="Ehrig, Hartmut
and Kreowski, Hans-J{\"o}rg
and Rozenberg, Grzegorz",
title="An algebraic theory of graph reduction",
booktitle="Graph Grammars and Their Application to Computer Science",
year="1991",
publisher="Springer Berlin Heidelberg",
address="Berlin, Heidelberg",
pages="70--83",
isbn="978-3-540-38395-6"
}

@INPROCEEDINGS{10353159,
  author={Carmesin, Johannes and Kurkofka, Jan},
  booktitle={2023 IEEE 64th Annual Symposium on Foundations of Computer Science (FOCS)}, 
  title={Canonical decompositions of 3-connected graphs}, 
  year={2023},
  volume={},
  number={},
  pages={1887-1920},
  doi={10.1109/FOCS57990.2023.00115}}

@article{GoldnerHarary,
title={Note on a smallest nonhamiltonian maximal planar graph},
author={Goldner, A. and Harary, F.},
journal={Bulletin Malaysian Mathematical Sciences Society},
volume={6},
pages={41--42},
year={1975}
}

@article{MA2024115486,
title = {Type-II Apollonian network: More robust and more efficient Apollonian network},
journal = {Chaos, Solitons and Fractals},
volume = {188},
pages = {115486},
year = {2024},
issn = {0960-0779},
doi = {https://doi.org/10.1016/j.chaos.2024.115486},
url = {https://www.sciencedirect.com/science/article/pii/S0960077924010385},
author = {Fei Ma and Jinzhi Ouyang and Haobin Shi and Wei Pan and Ping Wang}
}

@article{PhysRevLett.94.018702,
  title = {Apollonian Networks: Simultaneously Scale-Free, Small World, Euclidean, Space Filling, and with Matching Graphs},
  author = {Andrade, Jos\'e S. and Herrmann, Hans J. and Andrade, Roberto F. S. and da Silva, Luciano R.},
  journal = {Phys. Rev. Lett.},
  volume = {94},
  issue = {1},
  pages = {018702},
  numpages = {4},
  year = {2005},
  month = {Jan},
  publisher = {American Physical Society},
  doi = {10.1103/PhysRevLett.94.018702},
  url = {https://link.aps.org/doi/10.1103/PhysRevLett.94.018702}
}

@article{BENANTAR199585,
title = {Triangle graphs},
journal = {Applied Numerical Mathematics},
volume = {17},
number = {2},
pages = {85-96},
year = {1995},
issn = {0168-9274},
doi = {https://doi.org/10.1016/0168-9274(95)00011-I},
url = {https://www.sciencedirect.com/science/article/pii/016892749500011I},
author = {Messaoud Benantar and Uḡur Doḡrusöz and Joseph E Flaherty and Mukkai S Krishnamoorthy}
}

@Inbook{Buchi1990,
author="Buchi, J. Richard
and Landweber, Lawrence H.",
title="Solving Sequential Conditions by Finite-State Strategies",
bookTitle="The Collected Works of J. Richard B{\"u}chi",
year="1990",
publisher="Springer New York",
address="New York, NY",
pages="525--541",
isbn="978-1-4613-8928-6",
doi="10.1007/978-1-4613-8928-6_29",
url="https://doi.org/10.1007/978-1-4613-8928-6_29"
}

@book{DBLP:books/daglib/0020348,
  author       = {Christel Baier and
                  Joost{-}Pieter Katoen},
  title        = {Principles of model checking},
  publisher    = {{MIT} Press},
  year         = {2008},
  isbn         = {978-0-262-02649-9},
  bibsource    = {dblp computer science bibliography, https://dblp.org}
}

@inproceedings{Lics25,
author={Bozga, Marius and Iosif, Radu and Zuleger, Florian},
  booktitle={2025 40th Annual ACM/IEEE Symposium on Logic in Computer Science (LICS)}, 
  title={Regular Grammars for Sets of Graphs of Tree-Width 2}, 
  year={2025},
  volume={},
  number={},
  pages={704-717},
  doi={10.1109/LICS65433.2025.00059}}

@article{journals/corr/abs-2310-04764,
  author       = {Radu Iosif and
                  Florian Zuleger},
  title        = {Characterizations of Definable Context-Free Graphs},
  journal      = {Logical Methods in Computer Science},
  volume       = {22},
  year         = {2026},
  url          = {https://doi.org/10.48550/arXiv.2310.04764},
  doi          = {10.48550/ARXIV.2310.04764},
  eprinttype    = {arXiv},
  eprint       = {2310.04764}
}

@inproceedings{DBLP:conf/mfcs/DoczkalP18,
  author       = {Christian Doczkal and
                  Damien Pous},
  editor       = {Igor Potapov and
                  Paul G. Spirakis and
                  James Worrell},
  title        = {Treewidth-Two Graphs as a Free Algebra},
  booktitle    = {43rd International Symposium on Mathematical Foundations of Computer
                  Science, {MFCS} 2018, August 27-31, 2018, Liverpool, {UK}},
  series       = {LIPIcs},
  volume       = {117},
  pages        = {60:1--60:15},
  publisher    = {Schloss Dagstuhl - Leibniz-Zentrum f{\"{u}}r Informatik},
  year         = {2018},
  url          = {https://doi.org/10.4230/LIPIcs.MFCS.2018.60},
  doi          = {10.4230/LIPICS.MFCS.2018.60},
  bibsource    = {dblp computer science bibliography, https://dblp.org}
}

@article{DUFFIN1965303,
title = {Topology of series-parallel networks},
journal = {Journal of Mathematical Analysis and Applications},
volume = {10},
number = {2},
pages = {303-318},
year = {1965},
issn = {0022-247X},
doi = {https://doi.org/10.1016/0022-247X(65)90125-3},
url = {https://www.sciencedirect.com/science/article/pii/0022247X65901253},
author = {R.J Duffin}
}

@inproceedings{conf/icalp/Doumane0P24,
  author       = {Amina Doumane and
                  Samuel Humeau and
                  Damien Pous},
  editor       = {Karl Bringmann and
                  Martin Grohe and
                  Gabriele Puppis and
                  Ola Svensson},
  title        = {A Finite Presentation of Graphs of Treewidth at Most Three},
  booktitle    = {51st International Colloquium on Automata, Languages, and Programming,
                  {ICALP} 2024, July 8-12, 2024, Tallinn, Estonia},
  series       = {LIPIcs},
  volume       = {297},
  pages        = {135:1--135:18},
  publisher    = {Schloss Dagstuhl - Leibniz-Zentrum f{\"{u}}r Informatik},
  year         = {2024}
}

@inproceedings{DBLP:conf/icalp/Doumane22,
  author       = {Amina Doumane},
  editor       = {Mikolaj Bojanczyk and
                  Emanuela Merelli and
                  David P. Woodruff},
  title        = {Regular Expressions for Tree-Width 2 Graphs},
  booktitle    = {49th International Colloquium on Automata, Languages, and Programming,
                  {ICALP} 2022, July 4-8, 2022, Paris, France},
  series       = {LIPIcs},
  volume       = {229},
  pages        = {121:1--121:20},
  publisher    = {Schloss Dagstuhl - Leibniz-Zentrum f{\"{u}}r Informatik},
  year         = {2022},
  url          = {https://doi.org/10.4230/LIPIcs.ICALP.2022.121},
  doi          = {10.4230/LIPICS.ICALP.2022.121},
  bibsource    = {dblp computer science bibliography, https://dblp.org}
}

@article{Seese91,
title = {The structure of the models of decidable monadic theories of graphs},
journal = {Annals of Pure and Applied Logic},
volume = {53},
number = {2},
pages = {169-195},
year = {1991},
issn = {0168-0072},
doi = {https://doi.org/10.1016/0168-0072(91)90054-P},
url = {https://www.sciencedirect.com/science/article/pii/016800729190054P},
author = {D. Seese},
}

@book{DBLP:series/txtcs/FlumG06,
  author       = {J{\"{o}}rg Flum and
                  Martin Grohe},
  title        = {Parameterized Complexity Theory},
  series       = {Texts in Theoretical Computer Science. An {EATCS} Series},
  publisher    = {Springer},
  year         = {2006},
  url          = {https://doi.org/10.1007/3-540-29953-X},
  doi          = {10.1007/3-540-29953-X},
  isbn         = {978-3-540-29952-3},
  bibsource    = {dblp computer science bibliography, https://dblp.org}
}

@article{journals/lmcs/BojanczykP22,
  author       = {Mikolaj Bojanczyk and
                  Michal Pilipczuk},
  title        = {Optimizing tree decompositions in {MSO}},
  journal      = {Log. Methods Comput. Sci.},
  volume       = {18},
  number       = {1},
  year         = {2022},
  url          = {https://doi.org/10.46298/lmcs-18(1:26)2022},
  doi          = {10.46298/lmcs-18(1:26)2022},
  bibsource    = {dblp computer science bibliography, https://dblp.org}
}

@Inbook{Buechi90,
author="B{\"u}chi, J. Richard",
title="Weak Second-Order Arithmetic and Finite Automata",
bookTitle="The Collected Works of J. Richard B{\"u}chi",
year="1990",
publisher="Springer New York",
address="New York, NY",
pages="398--424",
isbn="978-1-4613-8928-6",
doi="10.1007/978-1-4613-8928-6_22",
url="https://doi.org/10.1007/978-1-4613-8928-6_22"
}

@article{Doner70,
author = {Doner, John},
title = {Tree Acceptors and Some of Their Applications},
year = {1970},
issue_date = {October, 1970},
publisher = {Academic Press, Inc.},
address = {USA},
volume = {4},
number = {5},
issn = {0022-0000},
url = {https://doi.org/10.1016/S0022-0000(70)80041-1},
doi = {10.1016/S0022-0000(70)80041-1},
journal = {J. Comput. Syst. Sci.},
month = {oct},
pages = {406–451},
numpages = {46}
}

@article{CourcelleI,
title = {The monadic second-order logic of graphs. I. Recognizable sets of finite graphs},
journal = {Information and Computation},
volume = {85},
number = {1},
pages = {12-75},
year = {1990},
issn = {0890-5401},
doi = {https://doi.org/10.1016/0890-5401(90)90043-H},
url = {https://www.sciencedirect.com/science/article/pii/089054019090043H},
author = {Bruno Courcelle},
}

@article{CourcelleV,
title = {The monadic second-order logic of graphs V: on closing the gap between definability and recognizability},
journal = {Theoretical Computer Science},
volume = {80},
number = {2},
pages = {153-202},
year = {1991},
issn = {0304-3975},
doi = {https://doi.org/10.1016/0304-3975(91)90387-H},
url = {https://www.sciencedirect.com/science/article/pii/030439759190387H},
author = {Bruno Courcelle}
}

@inproceedings{10.1145/2933575.2934508,
author = {Boja\'{n}czyk, Miko\l{}aj and Pilipczuk, Micha\l{}},
title = {Definability Equals Recognizability for Graphs of Bounded Treewidth},
year = {2016},
isbn = {9781450343916},
publisher = {Association for Computing Machinery},
address = {New York, NY, USA},
url = {https://doi.org/10.1145/2933575.2934508},
doi = {10.1145/2933575.2934508},
booktitle = {Proceedings of the 31st Annual ACM/IEEE Symposium on Logic in Computer Science},
pages = {407–416},
numpages = {10},
location = {New York, NY, USA},
series = {LICS '16}
}

@book{courcelle_engelfriet_2012,
place={Cambridge},
series={Encyclopedia of Mathematics and its Applications},
title={Graph Structure and Monadic Second-Order Logic: A Language-Theoretic Approach},
DOI={10.1017/CBO9780511977619},
publisher={Cambridge University Press},
author={Courcelle, Bruno and Engelfriet, Joost},
year={2012},
collection={Encyclopedia of Mathematics and its Applications}}

@article{ARNBORG19901,
title = {Forbidden minors characterization of partial 3-trees},
journal = {Discrete Mathematics},
volume = {80},
number = {1},
pages = {1-19},
year = {1990},
issn = {0012-365X},
doi = {https://doi.org/10.1016/0012-365X(90)90292-P},
url = {https://www.sciencedirect.com/science/article/pii/0012365X9090292P},
author = {Stefan Arnborg and Andrzej Proskurowski and Derek G. Corneil}
}

@article{SEYMOUR199322,
title = {Graph Searching and a Min-Max Theorem for Tree-Width},
journal = {Journal of Combinatorial Theory, Series B},
volume = {58},
number = {1},
pages = {22-33},
year = {1993},
issn = {0095-8956},
doi = {https://doi.org/10.1006/jctb.1993.1027},
url = {https://www.sciencedirect.com/science/article/pii/S0095895683710270},
author = {P.D. Seymour and R. Thomas},
}

\appendix
\section{Formal Definition of Transductions}
\label{app:transductions}

Let $\relations$ and $\relations'$ be relational signatures. A
relation $\trans$ between $\relations$- and $\relations'$-structures
is a $k$-\emph{copying} $(\relations,\relations')$-\emph{transduction}
if each output structure $\astruc' \in \trans(\astruc)$ is produced
from $k$ disjoint copies, called \emph{layers}, of the input structure
$\astruc$. The transduction is said to be \emph{copyless} if $k=1$.
The outcome of the transduction also depends on the valutation of zero
or more \emph{set parameters} $X_1, \ldots, X_n \in \Vars$, that range
over the subsets of the input universe. The transduction is said to be
\emph{parameterless} if $n=0$. Formally, we define
$(\relations,\relations')$-transductions using \emph{transduction
schemes}, i.e., finite tuples of \mso{} formul{\ae}:
\begin{align*}
  \scheme=\tuple{\varphi,\set{\psi_i}_{i \in \interv{1}{k}},
    \set{\theta_{(\qrel,i_1,\ldots,i_{\arityof{\qrel}})}}_{
      \qrel\in\relations',~
      i_1,\ldots,i_{\arityof{\qrel}}\in\interv{1}{k}
  }}
\end{align*}
where: \begin{compactitem}[-]
\item $\varphi(X_1,\ldots,X_n)$ defines the input structures
  $(\univ,\struc)$ for which the transduction has an output
  $(\univ',\struc')$, i.e., $(\univ,\struc) \models^\store \varphi$,
  for a store $\store$ that maps each $X_i$ into a set $\store(X_i)
  \subseteq \univ$,
\item $\psi_i(x_1, X_1, \ldots, X_n)$ defines the elements from the
  $i$-th layer copied in the output universe:
  \[\univ' \isdef \set{(u,i) \in \univ \times \interv{1}{k} \mid
    (\univ,\struc) \models^{\store[x_1 \leftarrow u]} \psi_i}\]
\item
  $\theta_{(\qrel,i_1,\ldots,i_{\arityof{\qrel}})}(x_1, \ldots,
  x_{\arityof{\qrel}}, X_1, \ldots, X_n)$ define the interpretation of
  $\qrel\in\relations'$ in the output:
  \[\struc'(\qrel) \isdef \{\tuple{(u_1,i_1), \ldots, (u_{\arityof{\qrel}}, i_{\arityof{\qrel}})} \mid
  (\univ,\struc) \Models^{\store[x_1\leftarrow u_1, \ldots, x_{\arityof{\qrel}} \leftarrow u_{\arityof{\qrel}}]}
  \theta_{(\qrel,i_1,\ldots,i_{\arityof{\qrel}})},~ i_1,\ldots,i_{\arityof{\qrel}} \in \interv{1}{k}\}\]
\end{compactitem}
For a given store $\store$, the output of the transduction is denoted
by $\defdof{\scheme}{\store} \isdef (\univ',\struc')$, where the
structure $(\univ',\struc')$ is the one defined above. Note that the
store valuations of $X_1, \ldots, X_n$ are the same everywhere in the
definition of $(\univ',\struc')$. The set $\defd{\scheme}(\astruc)$ is
the closure under isomorphism of the set
$\set{\defdof{\scheme}{\store} \mid \astruc \models^\store \varphi}$,
i.e., the output structures are the structures isomorphic to some
$\defdof{\scheme}{\store}$, whose elements are not necessarily pairs
of the form $(u,i)\in\univ\times\interv{1}{k}$. A transduction
$\trans$ is \emph{definable} iff $\trans=\defd{\scheme}$, for some
transduction scheme $\scheme$.

\end{document}